 \RequirePackage{fix-cm}

 \documentclass[smallextended]{svjour3} 
 \smartqed  

\usepackage{amssymb,amsmath,amsfonts,latexsym} 
\usepackage{hyperref}
   
\usepackage{mathtools}
\usepackage{multirow,array}
\usepackage{subcaption}
\usepackage{graphicx}
\usepackage{epstopdf}
\usepackage{float} 
\usepackage{color, xcolor}
\allowdisplaybreaks

\journalname{TEST} 

\begin{document}


\title{LRD spectral analysis of multifractional  functional time series on manifolds}

\titlerunning{LRD manifold multifractional  functional time series}

\author{
      Diana  P. Ovalle--Mu\~noz$^1$ 
\and
        M. Dolores Ruiz--Medina$^2$ 
 }

\authorrunning{D. P. Ovalle--Mu\~noz \and M. D. Ruiz--Medina} 

\institute{University of Granada$^{1,*,2}$
\at
 Faculty of Sciences, Avd, Fuente Nueva s/n, 18071 Granada, Spain\\
\email{mruiz@ugr.es} $^*$ corresponding author 
}

\date{Received: 1 December 2022 / Revised:  / Accepted: ......}

\maketitle

\begin{abstract} This paper addresses the estimation of the second--order structure of a manifold cross--time random field (RF) displaying spatially varying  Long Range Dependence (LRD), adopting the  functional time series framework  introduced  in  \cite{RuizMedina2022}.
Conditions for the asymptotic unbiasedness of the integrated periodogram operator in  the Hilbert--Schmidt operator norm are derived beyond structural assumptions. Weak--consistent  estimation of the long--memory operator is achieved under a semiparametric functional spectral framework in  the Gaussian context. The case where the projected manifold process  can display Short Range Dependence (SRD) and LRD at different  manifold scales is also analyzed.
 The performance of both estimation procedures is illustrated in the simulation study,  in the context of multifractionally integrated spherical functional autoregressive--moving average (SPHARMA(p,q))  processes.



\end{abstract} 


\section{Introduction} \label{sec:1}

\setcounter{section}{1} \setcounter{equation}{0} 
The literature on weakly dependent functional time series has been widely developed in the last few decades, allowing the statistical analysis,  and inference on stochastic  processes under a Markovian framework (see, e.g., \cite{Bosq2000}; \cite{Horvath2012}).  Nowadays, spectral analysis of functional time series models constitutes
an open research area.
 Under suitable functional cumulant mixing conditions, and the summability in time of the trace norm of the elements of the covariance operator family, in \cite{Panaretos13}, a weighted periodogram operator estimator  of the spectral density operator is derived. Its asymptotic analysis is addressed. Particularly,   the asymptotic normality of the functional discrete Fourier transform (fDFT) of the curve data is proved (see also \cite{Tavakoli2014}). In \cite{Panaretos13b},  a harmonic principal component analysis of functional time series, based on  Karhunen–Lo\'eve--like decomposition in the temporal functional spectral domain is proposed, the so–called Cram\'er–Karhunen–Lo\'eve representation  (see also \cite{Rubin20a};  \cite{Rubin20b}).   Some recent applications in the context of functional regression are obtained  in \cite{Pham2018}. Hypothesis testing for detecting modelling differences in functional time series dynamics  is achieved in \cite{Tavakoli2016} in the functional spectral domain.

Recently,  an attempt to extend spectral analysis of  functional time series   to the context of  LRD functional sequences has been presented  in \cite{RuizMedina2022}, covering, in particular,  some examples of the LRD funtional time series family analyzed in \cite{LiRobinsonShang19} in the temporal domain.   The application   of harmonic analysis in this more general context entails important  advantages as given in   \cite{RuizMedina2022}. Particularly, under stationary in time, the temporal dependence range can be approximated from the behavior in a neighborhood of zero frequency of the spectral density operator family  at different spatial resolution levels. Moreover, a more flexible modeling framework can be introduced in this setting. Particularly,  the projected process can display LRD and  SRD depending on the spatial scale, according to the support of the spectral measure of the LRD operator characterizing the distribution of its  eigenvalues.  In \cite{LiRobinsonShang19}, Functional Principal Component
Analysis (FPCA),  based on the long-run covariance function,  is applied in the consistent estimation of the dimension and the orthonormal
functions  spanning the dominant subspace, where the projected curve
process displays the largest dependence range. Fractionally integrated functional autoregressive moving averages processes  constitute an  interesting example (see \cite{LiRobinsonShang19}). The  multifractionally integrated version of this process family  can be analyzed under the modeling framework introduced in \cite{RuizMedina2022}.

Connected  and compact two--point homogeneous spaces constitute an example of manifold, with isometrically equivalent  properties to the sphere, locally resembles an Euclidean space. Here, we will denote it as   $\mathbb{M}_{d},$ with $d$ being  its topological dimension. The isotropy or   invariance of a kernel   with respect to the group of isometries of $\mathbb{M}_{d}$ allows its diagonal representation  in terms of a fixed  orthonormal basis given by  the eigenfunctions of the Laplace--Beltrami operator on $L^{2}(\mathbb{M}_{d}, d\nu).$ Thus, the separable Hilbert space $H=L^{2}\left(\mathbb{M}_{d},d\nu\right)$ of square integrable functions on  $\mathbb{M}_{d}$
  is considered  in our
 functional time series analysis. Here, $d\nu$ is the normalized Riemannian measure on $\mathbb{M}_{d}.$  Particularly, this Hilbert space framework has been adopted by several authors in the current literature for the special case of the sphere. That is the case  of \cite{CaponeraMarinucci}, where  estimation and asymptotic analysis of  spherical functional time series is achieved,  introducing  new model families (see \cite{Caponera21}).  Also, in the LRD framework,  sphere cross-time random ﬁelds are analyzed  in \cite{MarinucciRV}, investigating the asymptotic  behavior, under temporal increasing domain, of the empirical measure of  a excursion area  at any threshold. Time--dependent RF solution to a fractional pseudodifferential equation on the sphere is introduced in \cite{Ovidio16}  (see also \cite{Anh2018}). The eigenfunctions of the Laplace Beltrami operator on $L^{2}\left(\mathbb{M}_{d},d\nu\right),$  and the  corresponding zonal functions  play a crucial role in the  analysis of  manifold cross-time random ﬁelds (see, e.g.,  \cite{MaMalyarenko}). For example, Sobolev regularity and H\"older continuity of Gaussian RFs on a connected and compact two--point homogeneous space are  studied in  \cite{Cleanthous} and \cite{Cleanthous2}, by exploiting asymptotic properties of the pure point spectra of invariant kernels, and projection of functions into the eigenfunctions of the Laplace Beltrami operator.
Some motivating real data applications can be found  in the field of  Cosmic Microwave Background   (CMB) radiation  (see, e.g.,  \cite{Marinucci} and references therein).
 It is  well-known the interest of these RFs
 in  climatic change analysis (see, e.g., \cite{Alegria21}).

 The spectral domain allows to characterize LRD in functional time series in terms of the unboundedness at zero frequency of the corresponding element of the family of spectral density operators. Specifically, the divergence of the eigenvalues of the elements of the  spectral density operator family  at a neighborhood of zero frequency leads to different levels of singularity at this frequency depending on the spatial scale (see, e.g., \cite{RuizMedina2022}). The case of SRD and LRD at different spatial scales can also be analyzed, in the case of non--trivial null space  of the LRD operator. Thus,  the projected process displays SRD in this subspace, and LRD in the eigenspaces associated with the elements of the support of the spectral measure of the   LRD operator.

  From a theoretical point of view,  this paper contributes providing a sufficient condition for the asymptotic unbiasedness of the integrated  periodogram operator, in the Hilbert--Schmidt operator norm, for the class of   zero--mean, stationary and isotropic,  mean--square continuous Gaussian, or elliptically contoured, spatiotemporal RFs on $\mathbb{M}_{d}.$ This result provides a suitable setting for applying the LRD spectral  functional time series framework  introduced in  \cite{RuizMedina2022}, where  the weak--consistent estimation of the second--order structure  in the functional spectral domain is achieved  under a
     Gaussian scenario. Our  scenario is a bit different, since the Hilbert--Schmidt operator scenario has been considered.  Note that this scenario has usually been adopted in the current literature on functional time series  (see \cite{Bosq2000}; \cite{CaponeraMarinucci}). The case where
 SRD and LRD are displayed at different  manifold scales is also addressed, combining   semiparametric estimation of the spectral density operator based on minimum contrast    at manifold scales where LRD is displayed, and a nonparametric estimation based on the weighted periodogram  operator  in the remaining manifold scales where SRD is observed. The preliminary result derived in the Supplementary Material on spatiotemporal Karhunen--Lo\'eve expansion of the restriction to a bounded closed temporal interval of a   zero--mean, stationary in time and isotropic in space,  mean--square continuous Gaussian spatiotemporal RF on $\mathbb{M}_{d}$ has been applied.

From a practical point of view, the simulation study undertaken  in  Section \ref{sec:6} illustrates the performance beyond the Gaussian scenario of the two proposed estimation approaches in the functional spectral domain, in the context of     multifractionally integrated SPHARMA(p,q)   processes. Two cases are analyzed respectively corresponding to a decreasing and increasing positive bounded eigenvalue sequence of the  LRD operator. Particularly,  we study the scenario where  the projected process displays SRD and LRD   at different spherical scales. This last case is illustrated  when  the eigenvalues  of the LRD operator vanish at high discrete Legendre frequencies.    Details on the  implementation of both estimation methodologies, and some conclusions on the  results obtained   beyond the Gaussian scenario  are  also included in  Sections \ref{sec:6} --\ref{sec:7} (see also the Supplementary Material).

The outline of the paper is the following. Section \ref{sec:2} introduces notation, technical tools and preliminary elements.  Asymptotic unbiasedness of the integrated periodogram operator, and weak consistent minimum contrast estimation are studied in Section \ref{sec:3}.
  An extended formulation of this estimation methodology to the case where temporal LRD and SRD are displayed at different manifold scales is proposed     in Section \ref{sec:4}.  A simulation study is undertaken in Section \ref{sec:6} to illustrate the performance   of the proposed estimation methodologies. A summary of conclusions about the simulation study is given  in Section \ref{sec:7}. The Supplementary Material complements the theoretical material about the paper, and provides  the results of the simulation study in the remaining scenarios analyzed  that  are not displayed in the paper.

 \section{Preliminaries} \label{sec:2}
\setcounter{section}{2} \setcounter{equation}{0} 
\setcounter{section}{2} \setcounter{equation}{0} 
The family of manifold cross--time RFs analyzed in this paper is  introduced in this section. The  connected and compact two--point homogeneous spaces are briefly described, providing some preliminary algebraic notions, with reference to the invariant probabilistic measure, and the Laplace Beltrami operator.

Let $X=\{ X(\mathbf{x},t),\ \mathbf{x}\in \mathbb{M}_{d},\ t\in \mathbb{T}\}$ be a wide sense stationary in time and isotropic in space  zero--mean,  and mean--square continuous Gaussian, or elliptically contoured, RF on the basic probability space $(\Omega ,\mathcal{A},P).$   Here, $\mathbb{T}$ denotes the temporal domain, which usually is $\mathbb{Z}$ or $\mathbb{R}$ (see also the Supplementary Material for the case of a bounded temporal interval $[0,T]$).   Assume also that the map $\widetilde{X}_{t}:(\Omega ,\mathcal{A})\longrightarrow \left(L^{2}(\mathbb{M}_{d},d\nu),\mathcal{B}(L^{2}(\mathbb{M}_{d},d\nu))\right)$ is   measurable, with  $\widetilde{X}_{t}(\mathbf{x}):=X(\mathbf{x},t)$ for every $t\in \mathbb{T}$ and $\mathbf{x}\in \mathbb{M}_{d}.$  Here,  $\mathcal{B}(L^{2}(\mathbb{M}_{d},d\nu)) $ denotes the $\sigma$--algebra generated by all cylindrical subsets of $L^{2}(\mathbb{M}_{d},d\nu).$

Let $d_{\mathbb{M}_{d}}$ be the geodesic distance induced by the isometry with the unit sphere, preserving the spherical distance $\rho(\mathbf{x},\mathbf{y})=\cos^{-1}(\left\langle\mathbf{x},\mathbf{y}\right\rangle),$   $\mathbf{x},\mathbf{y}\in \mathbb{S}_{d},$
and  let $\omega_{d}=\int_{\mathbb{M}_{d}}d\nu(\mathbf{x}).$ In the following,
$R_{n}^{(\alpha, \beta )}\left(\cos\left(d_{\mathbb{M}_{d}}(\mathbf{x},\mathbf{y})\right)\right)=$ \linebreak $\frac{P_{n}^{(\alpha, \beta )}\left(\cos\left(d_{\mathbb{M}_{d}}(\mathbf{x},\mathbf{y})\right)\right)}{P_{n}^{(\alpha, \beta )}\left(1\right)},$  with $P_{n}^{(\alpha, \beta )}$ being  the Jacobi polynomial of degree $n$ with a pair of parameters $(\alpha ,\beta )$  (see, e.g., \cite{Andrews99}).
For each $n\in \mathbb{N}_{0},$ $\{ S_{n,1}^{d},\dots,S_{n,\delta(n,d)}^{d}\}$   is the orthonormal basis of eigenfunctions of the
eigenspace  $H_{n}$ of Laplace--Beltrami operator $\Delta_{d}$  on $L^{2}(\mathbb{M}_{d},d\nu),$  associated with the eigenvalue
 $\lambda_{n}=-n\varepsilon(n\varepsilon+\alpha +\beta +1).$  The dimension $\delta(n,d)$ of
$H_{n}$ is given by  $$\delta(n,d)=\frac{(2n+\alpha+\beta +1)\Gamma (\beta +1)\Gamma (n+\alpha +\beta +1)\Gamma (n+\alpha +1)}{\Gamma (\alpha +1)\Gamma (\alpha +\beta +2)\Gamma (n+1)\Gamma (n+\beta +1)}.$$
Note that $\alpha =(p+q-1)/2,$  $\beta =(q-1)/2,$  and $\varepsilon =2$ if $\mathbb{M}_{d}=\mathbb{P}^{d}(\mathbb{R}),$ the projective space over the field $\mathbb{R},$   and  $\varepsilon =1,$ otherwise.
Parameters $q$ and $p$ are the dimensions of some root spaces connected with   the Lie algebras of the groups $G$ and $K,$   with $\mathbb{M}_{d}\simeq G/K,$ being $G$ the connected component of the group of isometries of $\mathbb{M}_{d},$
 and $K$ the stationary subgroup of a fixed point \textbf{o} in $\mathbb{M}_{d}$
 (see, e.g., Table 1  in  \cite{MaMalyarenko}).

The next lemma is applied along the paper.

\begin{lemma}
\label{cor:adformula}
 (See \cite[Theorem 3.2.]{Gine1975} and \cite[p 455]{Andrews99})
For every $n\in \mathbb{N}_{0},$ the following addition formula holds:
\begin{equation}
\sum_{j=1}^{\delta(n,d)}S_{n,j}^{d}(\mathbf{x})S_{n,j}^{d}(\mathbf{y})=\frac{\delta(n,d)}{\omega_{d}}R_{n}^{\alpha,\beta }\left(\cos(d_{\mathbb{M}_{d}}(\mathbf{x},\mathbf{y}))\right),\quad \mathbf{x},\mathbf{y}\in \mathbb{M}_{d}.
\label{af}
\end{equation}
\end{lemma}

Let $C(d_{\mathbb{M}_{d}}(\mathbf{x},\mathbf{y}),t-s)= E\left[X(\mathbf{x},t) X(\mathbf{y},s)\right],$ for $\mathbf{x},\mathbf{y}\in \mathbb{M}_{d},$ and $t,s\in \mathbb{T},$ be the covariance function of $X.$  Assume that  $C(d_{\mathbb{M}_{d}}(\mathbf{x},\mathbf{y}),t)=C(d_{\mathbb{M}_{d}}(\mathbf{x},\mathbf{y}),-t).$
The following diagonal series expansion  holds  under the conditions of   Theorem  4   in  \cite{MaMalyarenko}:
 \begin{eqnarray}&&C(d_{\mathbb{M}_{d}}(\mathbf{x},\mathbf{y}),t-s)=
  \sum_{n\in \mathbb{N}_{0}}
B_{n}(t-s)\sum_{j=1}^{\delta (n,d)}S_{n,j}^{d}(\mathbf{x})S_{n,j}^{d}(\mathbf{y})\nonumber\\
&&=
\sum_{n\in \mathbb{N}_{0}} \frac{\delta (n,d)}{\omega_{d}}  B_{n}(t-s)R_{n}^{(\alpha, \beta )}\left(\cos\left(d_{\mathbb{M}_{d}}(\mathbf{x},\mathbf{y})\right)\right),\ \mathbf{x},\mathbf{y}\in \mathbb{M}_{d}, \ t,s \in \mathbb{T}.\nonumber\\ \label{klexpc2}
\end{eqnarray}

 \section{Operator based minimum contrast parameter estimation of LRD} \label{sec:3}
\setcounter{section}{3} \setcounter{equation}{0} 
This section  adopts the semiparametric spectral LRD  functional time series framework introduced in \cite{RuizMedina2022} (see \textbf{Condition C1} below), for inference on an LRD manifold cross–time RF, constructed from a zero–mean, stationary in time, and isotropic in space, mean–square continuous Gaussian,
or elliptically contoured spatiotemporal RF  $X=\{ X(\mathbf{x},t),\ \mathbf{x}\in \mathbb{M}_{d},\ t\in \mathbb{T}=\mathbb{Z}\}$   (see, e.g., \cite{MarinucciRV} for sphere cross–time RFs).

 Note that under the conditions assumed in Section \ref{sec:2},
 $X=\{ X(\mathbf{x},t),\ \mathbf{x}\in \mathbb{M}_{d},\ t\in \mathbb{T}=\mathbb{Z}\}$ defines a functional time series $\{ \widetilde{X}_{t}(\cdot ),\  t\in \mathbb{Z}\}$. In Section \ref{mce}
 below, we will work under this scenario to implement minimum contrast parameter estimation under  \textbf{Condition C1} (providing  the semiparametric functional spectral framework  introduced in \cite{RuizMedina2022}).

  \textbf{Condition C0} below establishes a sufficient condition for the asymptotic  unbiasedness of the integrated periodogram operator of $X$ in the Hilbert \linebreak -Schmidt operator norm beyond structural assumptions (see Theorem \ref{th1}). The weak--consistency of the minimum contrast estimator of the LRD operator is then obtained in  Theorem \ref{th2}, under new \textbf{Condition C0} and  \textbf{Condition C1}, adopting the
$L^{2}(\mathbb{M}_{d}, d\nu )$--valued time series framework.

  The following condition will be assumed in the subsequent development.

\medskip

\noindent \textbf{Condition C0}.  The elements of the temporal coefficient sequence  $\left\{B_{n}(\tau ),\right.$ $\left.n\in \mathbb{N}_{0}\right\}$ in equation (\ref{klexpc2}) are such that  \begin{equation}
\sum_{n\in \mathbb{N}_{0}}\delta (n,d)\sum_{\tau \in \mathbb{Z}}B_{n}^{2}(\tau )<\infty.
\label{hsc}
\end{equation}

 Under condition (\ref{hsc}) (see equation (\ref{klexpc2c}) below),
 one can define almost surely (a.s.) in $\omega \in [-\pi,\pi],$ i.e., $\omega \in [-\pi,\pi]\backslash \mathcal{D}_{0},$ $\int_{\mathcal{D}_{0}}d\omega=0,$
 the spectral density operator $\mathcal{F}_{\omega}$ in the space  $\mathcal{S}(L^{2}(\mathbb{M}_{d},d\nu; \mathbb{C}))$ of Hilbert--Schmidt operators on $L^{2}(\mathbb{M}_{d},d\nu; \mathbb{C})$
 as follows:
   \begin{equation}
  \mathcal{F}_{\omega}
  \underset{\mathcal{S}(L^{2}(\mathbb{M}_{d},d\nu; \mathbb{C}))}{=}
  \frac{1}{2\pi} \sum_{\tau \in \mathbb{Z}}\exp\left(-i\omega \tau\right)\mathcal{R}_{\tau},\label{sdo2}
\end{equation}
\noindent where  $\mathcal{R}_{\tau}=E[\widetilde{X}_{s}\otimes \widetilde{X}_{s+\tau}]=E[ \widetilde{X}_{s+ \tau}\otimes\widetilde{X}_{s}],$ for $\tau,s\in \mathbb{Z}.$

 Note that Parseval identity  leads  to
\begin{eqnarray}&&
\int_{-\pi}^{\pi}\|\mathcal{F}_{\omega }\|_{\mathcal{S}(L^{2}(\mathbb{M}_{d},d\nu; \mathbb{C}))}^{2}d\omega
=\sum_{n\in \mathbb{N}_{0}}\delta (n,d)\int_{-\pi}^{\pi}|f_{n}(\omega )|^{2}d\omega
\nonumber\\
&&=\sum_{n\in \mathbb{N}_{0}}\delta (n,d)\sum_{\tau \in \mathbb{Z}}B_{n}^{2}(\tau )
=\sum_{\tau \in \mathbb{Z}}\left\|\mathcal{R}_{\tau }\right\|_{\mathcal{S}(L^{2}(\mathbb{M}_{d},d\nu; \mathbb{R}))}^{2}
<\infty.
\label{klexpc2c}
\end{eqnarray}

 From equations   (\ref{klexpc2}) and   (\ref{klexpc2c}),
     $\mathcal{F}_{\omega}$ has  kernel $f_{\omega }(\mathbf{x},\mathbf{y}),$  $\mathbf{x}, \mathbf{y}\in \mathbb{M}_{d},$  admitting  for $\omega \in [-\pi,\pi]\backslash \mathcal{D}_{0},$ the following diagonal series expansion in the space $\mathcal{S}(L^{2}(\mathbb{M}_{d},d\nu; \mathbb{C}))\equiv  L^{2}(\mathbb{M}_{d}\times \mathbb{M}_{d},d\nu\otimes d\nu; \mathbb{C}):$
 For $\mathbf{x}, \mathbf{y}\in \mathbb{M}_{d},$
\begin{eqnarray}&&
f_{\omega }(\mathbf{x},\mathbf{y})\underset{L^{2}(\mathbb{M}_{d}\times \mathbb{M}_{d})}{=}\frac{1}{2\pi}\sum_{n\in \mathbb{N}_{0}}f_{n}(\omega )\sum_{j=1}^{\delta (n,d)}S_{n,j}^{d}(\mathbf{x})S_{n,j}^{d}(\mathbf{y})
\nonumber\\&&=\frac{1}{2\pi}
\sum_{n\in \mathbb{N}_{0}}
 \left[\sum_{\tau \in \mathbb{Z}}\exp(-i\omega \tau)B_{n}(\tau) \right]\frac{\delta (n,d)}{\omega_{d}}R_{n}^{(\alpha, \beta )}\left(\cos\left(d_{\mathbb{M}_{d}}(\mathbf{x},\mathbf{y})\right)\right).
\label{eigsdo}
\end{eqnarray}

  The convergence to zero in the Hilbert--Schmidt operator norm of the integrated bias of the  periodogram operator also  holds  under   (\ref{hsc}).
\begin{theorem}
\label{th1}
 Under    condition (\ref{hsc}),
$$
\left\|\int_{-\pi}^{\pi }\left[\mathcal{F}_{\omega}
-\mathcal{F}_{\omega}^{(T)}\right]d\omega\right\|_{\mathcal{S}\left(L^{2}(\mathbb{M}_{d},d\nu; \mathbb{C})\right)}\to 0,\quad  T\to \infty,
$$
\noindent where $\mathcal{F}_{\omega}^{(T)}$ denotes the mean of the periodogram operator
$$ p_{\omega }^{(T)}=\widetilde{X}_{\omega}^{(T)}\otimes
\overline{\widetilde{X}_{\omega}^{(T)}},\quad \omega \in [-\pi,\pi],$$
\noindent with $\widetilde{X}^{(T)}_{\omega }(\cdot)
\underset{L^{2}(\mathbb{M}_{d},d\nu; \mathbb{C})}{=} \frac{1}{\sqrt{2\pi
T}}\sum_{t=1}^{T}\widetilde{X}_{t}(\cdot )\exp\left(-i\omega t\right),$
$\omega\in  [-\pi,\pi],$ being  the functional Discrete Fourier Transform  (fDFT),
based on a functional sample $\widetilde{X}_{t},$ $t=1,\dots,T,$ of spatiotemporal RF $X.$ Here, as before, for each $t=1,\dots,T,$
  $\widetilde{X}_{t}(\mathbf{x}):= X(\mathbf{x},t),$ for every $\mathbf{x}\in \mathbb{M}_{d}.$
\end{theorem}

\begin{proof}
The proof follows from equation  (\ref{klexpc2c}) implying Lemma 1 in \cite{RuizMedina2022} holds in our context.
  Hence,  equations  (4.6)--(4.8) in  the proof of Theorem 1 in \cite{RuizMedina2022}, and the remaining steps of the proof  of this theorem can be obtained in a similar way.
   \end{proof}
\subsection{Minimum contrast parameter estimation}
\label{mce}

For the implementation of the minimum contrast parameter estimation of the spatial--varying  LRD  parameter of $X,$ in the $L^{2}(\mathbb{M}_{d},d\nu)$--valued time series   framework introduced  in \cite{RuizMedina2022},  the following condition is assumed:

\medskip
 \noindent \textbf{Condition C1}. The
 elements of the function sequence  of Fourier transforms  $\{f_{n}(\cdot ),\ n\in \mathbb{N}_{0}\}$  in (\ref{eigsdo})    admit the following semiparametric modeling: For every  $\omega \in [-\pi,\pi]\backslash \{0\},$
\begin{eqnarray}f_{n,\theta }(\omega )&=&
B_{n}^{\eta }(0)M_{n}(\omega )\left[4(\sin(\omega /2))^{2}\right]^{-\alpha (n,\theta)/2},\ \theta \in \Theta,\ n\in \mathbb{N}_{0},\label{eqsmc1}\end{eqnarray}
\noindent where   $l_{\alpha }(\theta )\leq \alpha (n,\theta)\leq L_{\alpha }(\theta ),$  for any $n\geq 0,$ and $\theta \in \Theta,$  for certain $l_{\alpha }(\theta ), L_{\alpha }(\theta )\in (0,1/2).$ The elements of the function sequence $\{M_{n},\ n\in \mathbb{N}_{0}\}$  are strictly positive  continuous functions on $[-\pi,\pi],$  slowly varying    at zero frequency  in the Zygmund’s sense (see, e.g., Definition 6.6 in \cite{Beran17}). For every $\omega \in [-\pi,\pi],$ $\left\{M_{n}(\omega ),\ n\in \mathbb{N}_{0}\right\}$   are the eigenvalues of   operator $\mathcal{M}_{\omega} \in \mathcal{S}(L^{2}(\mathbb{M}_{d},d\nu; \mathbb{C})),$  and $\|\mathcal{M}_{\omega}\|_{\mathcal{S}(L^{2}(\mathbb{M}_{d},d\nu; \mathbb{C}))}\in L^{2}([-\pi,\pi]).$
The  operator family $\left\{\mathcal{M}_{\omega},\ \omega \in [-\pi,\pi]\right\}$ defines  the SRD spectral family. The  associated kernel family satisfies
\begin{eqnarray}&&\mathcal{K}_{\mathcal{M}_{\omega}}(\mathbf{x},\mathbf{y})=\sum_{n\in \mathbb{N}_{0}} M_{n}(\omega )
\sum_{j=1}^{\delta (n,d)}S_{n,j}^{d}(\mathbf{x})S_{n,j}^{d}(\mathbf{y})\nonumber\\
&=&
\sum_{n\in \mathbb{N}_{0}}
M_{n}(\omega )\frac{\delta (n,d)}{\omega_{d}}R_{n}^{(\alpha, \beta )}\left(\cos\left(d_{\mathbb{M}_{d}}(\mathbf{x},\mathbf{y})\right)\right),\ \mathbf{x},\mathbf{y}\in \mathbb{M}_{d}.\nonumber
\end{eqnarray}

 The parameter space  $\Theta $ is assumed to be  compact   with non null interior.
For each $\theta \in \Theta,$ $\{\alpha (n,\theta),\ n\in \mathbb{N}_{0}\}$  is the  system of eigenvalues of  the parameterized long--memory operator  $\mathcal{A}_{\theta }$ with kernel $\mathcal{K}_{\mathcal{A}_{\theta }}$  admitting the following  diagonal series expansion:
\begin{eqnarray}&&
\mathcal{K}_{\mathcal{A}_{\theta }}(\mathbf{x},\mathbf{y})=\sum_{n\in \mathbb{N}_{0}} \alpha (n,\theta)
\sum_{j=1}^{\delta (n,d)}S_{n,j}^{d}(\mathbf{x})S_{n,j}^{d}(\mathbf{y})\nonumber\\
&&=
\sum_{n\in \mathbb{N}_{0}}
\alpha (n,\theta)\frac{\delta (n,d)}{\omega_{d}}R_{n}^{(\alpha, \beta )}\left(\cos\left(d_{\mathbb{M}_{d}}(\mathbf{x},\mathbf{y})\right)\right),\ \mathbf{x},\mathbf{y}\in \mathbb{M}_{d}.
\nonumber\end{eqnarray}
\noindent Hence,  for every $\theta \in \Theta,$ $\mathcal{A}_{\theta }$  has  isotropic kernel and defines  a strictly positive self--adjoint operator on $L^{2}(\mathbb{M}_{d},d\nu ),$ with  norm in the space $\mathcal{L}(L^{2}(\mathbb{M}_{d},d\nu))$  of bounded linear operators
 less than $1/2.$  Note that,  since  $\sin(\omega )\sim \omega,$
$\omega \to 0,$
\begin{equation}
\left|1-\exp\left(-i\omega \right)\right|^{-\mathcal{A}_{\theta
}}=[4\sin^{2}(\omega /2)]^{-\mathcal{A}_{\theta }/2}\sim |\omega
|^{-\mathcal{A}_{\theta }},\quad \omega \to 0,\label{ffsvint}
\end{equation}
\noindent where the frequency varying operator
  $\left|1-\exp\left(-i\omega \right)\right|^{-\mathcal{A}_{\theta }/2}$ is interpreted
as  in \cite{Charac14}; \cite{Rackauskasv2}. In particular, equation (3.1) in \textbf{Assumption II} in \cite{RuizMedina2022}, characterizing LRD of functional time series in the spectral domain, holds.
\begin{remark}\label{ac} \textbf{Condition C1} restricts the eigenvalues of $\mathcal{A}_{\theta }$  to the interval $(0, 1/2),$  leading to  a shorter range of spectral singularity  at zero frequency at  any spatial scale, allowing (\ref{klexpc2c}) holds. In particular,
\begin{eqnarray}
&&\int_{-\pi}^{\pi} \left\|\mathcal{F}_{\omega }\right\|_{\mathcal{S}(L^{2}(\mathbb{M}_{d},d\nu; \mathbb{C}))}^{2}d\omega=\sum_{n\in \mathbb{N}_{0}}\delta (n,d)\int_{-\pi}^{\pi}|f_{n}(\omega )|^{2}d\omega \nonumber\\ &&\hspace*{0.3cm} \leq
\left\{\left[\int_{-\pi}^{-1}+\int_{1}^{\pi}\right]|\omega|^{-2L(\theta ) } +  \int_{-1}^{1}  |\omega|^{-2l(\theta ) }\right\}
\left\|\mathcal{M}_{\omega }\right\|^{2}_{\mathcal{S}\left(L^{2}(\mathbb{M}_{d},d\nu; \mathbb{C})\right)}d\omega  <\infty.
\nonumber\\
\label{sqishsdo}
\end{eqnarray}
\end{remark}
 The sequence $\left\{B_{n}^{\eta }(0)\geq 0,\ n\in \mathbb{N}_{0}\right\}$  in equation (\ref{eqsmc1}) defines the eigenvalues of the trace integral autocovariance operator
 $\mathcal{R}^{\eta }_{0}=E\left[ \eta_{t} \otimes \eta_{t}  \right]=E\left[ \eta_{0} \otimes \eta_{0}  \right],$  $t\in \mathbb{Z},$  of the zero--mean  innovation process $\left\{\eta_{t},\ t\in \mathbb{Z}\right\},$ with kernel $r_{\eta }$ satisfying
 $$r_{\eta }(\mathbf{x},\mathbf{y})=\sum_{n\in \mathbb{N}_{0}}B_{n}^{\eta }(0)\sum_{j=1}^{\delta (n,d)}S_{n,j}^{d}(\mathbf{x})S_{n,j}^{d}(\mathbf{y}),\quad\forall \mathbf{x},\mathbf{y}\in \mathbb{M}_{d},
 $$
 \noindent where $\frac{B_{n}(0)}{B_{n}^{\eta }(0)}=\int_{-\pi}^{\pi}M_{n}(\omega )\left[4(\sin(\omega /2))^{2}\right]^{-\alpha (n,\theta)/2}d\omega,$  for each  $n\in \mathbb{N}_{0}.$

    Assume that the true parameter value $\theta_{0}$
lies in the non empty interior of compact set  $\Theta,$ and
 $\alpha (\cdot ,\theta_{1})\neq \alpha (\cdot ,\theta _{2}),$
 for $\theta_{1}\neq \theta_{2},$  and
 $\theta_{1},\theta_{2}\in \Theta, $ ensuring identifiability.
 Let  $\left\{\widehat{\alpha}_{T} (n,\theta )= \alpha (n, \widehat{\theta }_{T}),\ n\in \mathbb{N}_{0}\right\}$ be the  parametric estimators of the eigenvalues
 of $\mathcal{A}_{\theta }$ by minimum contrast.  The computation of the minimum contrast parametric estimator  $\widehat{\theta }_{T}$ requires the introduction of some operator families in the spectral domain as  now briefly describe.

 Operator integrals   are understood here as improper operator Stieltjes integrals which strongly converge  (see, e.g., Section 8.2.1 in  \cite{Ramm05}).
Specifically, in the definition of our loss function, integration in the temporal spectral domain with respect to  weighting operator $\mathcal{W}_{\omega}$ is achieved.  For each $\omega \in [-\pi,\pi],$  the invariant kernel  $\mathcal{K}_{\mathcal{W}_{\omega}}$   of   $\mathcal{W}_{\omega}$    on $\mathbb{M}_{d}\times \mathbb{M}_{d}$ admits the following series expansion:
\begin{eqnarray}
&& \mathcal{K}_{\mathcal{W}_{\omega}}(\mathbf{x},\mathbf{y})=\sum_{n\in \mathbb{N}_{0}} W(\omega ,n ,\gamma )
\sum_{j=1}^{\delta (n,d)}S_{n,j}^{d}(\mathbf{x})S_{n,j}^{d}(\mathbf{y})
\nonumber\\
&&=\sum_{n\in \mathbb{N}_{0}}
\widetilde{W}(n) |\omega |^{\gamma
}\sum_{j=1}^{\delta (n,d)}S_{n,j}^{d}(\mathbf{x})S_{n,j}^{d}(\mathbf{y}),\  \gamma>0,\ \mathbf{x},\mathbf{y}\in \mathbb{M}_{d}, \label{eqdeftildew}
\end{eqnarray}
\noindent where the  eigenvalues $\{ W(\omega ,n ,\gamma ) ,\ n\in \mathbb{N}_{0}\}$ factorize as
$W(\omega ,n ,\gamma )=
\widetilde{W}(n) |\omega |^{\gamma },$ for every $n\in \mathbb{N}_{0},$  with $\left\{\widetilde{W}(n),\ n\in \mathbb{N}_{0}\right\}$ defining the eigenvalues of a self--adjoint positive bounded operator
$\widetilde{\mathcal{W}} ,$ such that, for certain $m_{\widetilde{\mathcal{W}}}>0,$
$M_{\widetilde{\mathcal{W}}}>0$
\begin{eqnarray}&&
m_{\widetilde{\mathcal{W}}}\leq  \left\|\widetilde{\mathcal{W}}^{1/2}(\psi)\right\|^{2}_{L^{2}(\mathbb{M}_{d},d\nu; \mathbb{C})}
\leq M_{\widetilde{\mathcal{W}}},\ \psi\in L^{2}(\mathbb{M}_{d},d\nu; \mathbb{C});\
   \|\psi\|_{L^{2}(\mathbb{M}_{d},d\nu; \mathbb{C})}=1.
\nonumber
\end{eqnarray}
\noindent
Under \textbf{Condition C1}, define the normalizing  self--adjoint integral  operator  $\mathcal{N}_{\theta }$ by the  kernel
\begin{eqnarray}&&
\mathcal{K}_{\mathcal{N}_{\theta }}(\mathbf{x},\mathbf{y})
=\sum_{n\in \mathbb{N}_{0}}\widetilde{W}(n)\left[\int_{-\pi}^{\pi} \frac{B_{n}^{\eta }(0)M_{n}(\omega )\left[4(\sin(\omega /2))^{2}\right]^{-\alpha (n,\theta)/2}
}{|\omega|^{-\gamma }} d\omega \right]\nonumber\\
&&\hspace*{0.5cm}\times
\sum_{j=1}^{\delta (n,d)}S_{n,j}^{d}(\mathbf{x})S_{n,j}^{d}(\mathbf{y}),\quad \mathbf{x},\mathbf{y}\in \mathbb{M}_{d},\ \theta \in \Theta,\quad \gamma>0.
\label{eqdeftildew2}
\end{eqnarray}

Hence,  for each $\theta \in \Theta,$  and $\omega \in [-\pi,\pi],$  $\omega \neq 0,$ the kernel  $\mathcal{K}_{\Upsilon_{\omega ,\theta }}$ of the density operator
\begin{equation}
 \Upsilon_{\omega ,\theta }=[\mathcal{N}_{\theta }]^{-1}\mathcal{F}_{\omega ,\theta }= \mathcal{F}_{\omega ,\theta }[\mathcal{N}_{\theta }]^{-1}
\label{fsdo}
\end{equation}
\noindent satisfies
$
\int_{-\pi}^{\pi}\int_{\mathbb{M}_{d}}\mathcal{K}_{\Upsilon_{\omega ,\theta }}(\mathbf{x},\mathbf{y})\mathcal{K}_{\mathcal{W}_{\omega }}(\mathbf{y},\mathbf{z})
d\nu(\mathbf{y})d\omega =\delta (\mathbf{x}-\mathbf{z}),$ for $\mathbf{x},\mathbf{z}\in \mathbb{M}_{d},$
 in the  weak sense,  meaning that,
  for every $\varrho, \psi \in L^{2}(\mathbb{M}_{d},d\nu; \mathbb{C}),$
\begin{eqnarray}
&&
\int_{-\pi}^{\pi}\Upsilon_{\omega ,\theta }\mathcal{W}_{\omega }
(\varrho)(\psi )d\omega
=\left\langle
\varrho,\psi\right\rangle_{L^{2}(\mathbb{M}_{d},d\nu; \mathbb{C})},\quad \forall \theta \in \Theta.
\label{identityintegralop}
\end{eqnarray}
\noindent Here, $\delta (\mathbf{x}-\mathbf{y})$ denotes the Dirac Delta distribution.  Equivalently,
$\int_{-\pi}^{\pi}\Upsilon_{\omega ,\theta }\mathcal{W}_{\omega
}d\omega $  defines the identity operator  $\mathcal{I}_{L^{2}(\mathbb{M}_{d},d\nu; \mathbb{C})}$
on $L^{2}(\mathbb{M}_{d},d\nu; \mathbb{C})$ for all $\theta \in \Theta,$ having unitary eigenvalues, i.e.,  $\int_{-\pi}^{\pi}\Upsilon (\omega , n ,\theta
)W(n,\omega ,\gamma )d\omega=1,$ for $n\in \mathbb{N}_{0},$   $\theta \in \Theta .$

Given the candidate set constituted by the parametric operator  families \linebreak $\left\{\Upsilon_{\omega ,\theta },\ \omega \in [-\pi,\pi]\right\},$ $\theta \in \Theta ,$ the theoretical  loss function  $L(\theta_{0},\theta )$ to be minimized is defined  as
\begin{eqnarray}
&&L(\theta_{0},\theta ):=\|U_{\theta }-U_{\theta_{0} }\|_{\mathcal{L}\left(L^{2}(\mathbb{M}_{d},d\nu; \mathbb{C})\right)}\nonumber\\
&&=\left\|\int_{-\pi}^{\pi}\mathcal{F}_{\omega ,\theta_{0}}
\ln\left(\Upsilon_{\omega , \theta_{0} }\Upsilon_{\omega ,\theta
}^{-1}\right)\mathcal{W}_{\omega }d\omega \right\|_{\mathcal{L}\left(L^{2}(\mathbb{M}_{d},d\nu; \mathbb{C})\right)}
\nonumber\\
&&=\left\{\sup_{n\geq 0}\left|\widetilde{W}(n)\int_{-\pi}^{\pi}\frac{f_{n,\theta_{0} }(\omega )}{ |\omega|^{-\gamma
}}\ln\left(\frac{\Upsilon (\omega , n, \theta_{0})}{\Upsilon (\omega , n, \theta )
}\right) d\omega
\right|\right\},\label{eqkld}
\end{eqnarray}
\noindent where  the theoretical contrast operator $U_{\theta }$ has kernel
\begin{eqnarray}
&&  \mathcal{K}_{U_{\theta }}\left(\mathbf{x},\mathbf{y}\right)
=-\sum_{n\in \mathbb{N}_{0}}U_{\theta }(n)\left[\sum_{j=1}^{\delta (n,d)}S_{n,j}^{d}(\mathbf{x})S_{n,j}^{d}(\mathbf{y})\right]
\nonumber\\
&&
=-\sum_{n\in \mathbb{N}_{0}}
\left[\sum_{j=1}^{\delta (n,d)}S_{n,j}^{d}(\mathbf{x})S_{n,j}^{d}(\mathbf{y})\right]\int_{-\pi}^{\pi} \frac{B_{n}^{\eta }(0)M_{n}(\omega )}{\left[4(\sin(\omega /2))^{2}\right]^{\alpha (n,\theta_{0})/2}}
\nonumber\\
&&\hspace*{1.5cm}\times \ln\left(\Upsilon (\omega , n ,\theta
)\right)|\omega|^{\gamma }\widetilde{W}(n)d\omega ,\ \mathbf{x},\mathbf{y}\in \mathbb{M}_{d},\ \theta \in \Theta .\nonumber
\end{eqnarray}
 \noindent Hence, $U_{\theta }\in \mathcal{L}\left(L^{2}(\mathbb{M}_{d},d\nu; \mathbb{C})\right)$ (see also Remark 7 in \cite{RuizMedina2022}).

  For every $\theta \in \Theta ,$          the eigenvalues $\left\{L_{n}(\theta_{0}, \theta ),\ n\in \mathbb{N}_{0}\right\}$ of
 the loss operator  $U_{\theta }-U_{\theta_{0} }$ satisfy   $L_{n}(\theta_{0}, \theta )=\widetilde{W}(n)\int_{-\pi}^{\pi}\frac{f_{n,\theta_{0} }(\omega )}{ |\omega|^{-\gamma
}}\ln\left(\frac{\Upsilon (\omega , n, \theta_{0})}{\Upsilon (\omega , n, \theta )
}\right) d\omega\geq 0$ (see \cite{RuizMedina2022} for more details).
   Hence, \begin{eqnarray}
 && L(\theta_{0},\theta )=\sup_{n\geq 0}L_{n}(\theta_{0}, \theta )>0,\quad \theta  \neq  \theta_{0}\nonumber\\
  && L(\theta_{0},\theta )=\sup_{n\geq 0}L_{n}(\theta_{0}, \theta )=0\ \Leftrightarrow \  \theta =  \theta_{0}.
  \label{eqjibbh}
  \end{eqnarray}

  From  (\ref{eqjibbh}),
\begin{eqnarray}
\theta_{0}&=&
\mbox{arg} \
 \min_{\theta \in \Theta } \sup_{n\geq 0}L_{n}(\theta_{0}, \theta )
=\mbox{arg} \ \min_{\theta \in \Theta } \sup_{n\geq 0}U_{\theta }(n).\label{mfethetabb}\end{eqnarray}
The  empirical contrast operator  $\mathbf{U}_{T, \theta}$
\begin{equation}
\mathbf{U}_{T, \theta}=-\int_{-\pi}^{\pi }p_{\omega }^{(T)}
\ln\left(\Upsilon_{\omega ,\theta }\right)\mathcal{W}_{\omega} d\omega ,\ \theta\in \Theta,
\label{eco}
\end{equation}
\noindent then provides the empirical loss function, and  $\widehat{\theta }_{T}$ can be  computed as
\begin{eqnarray}
\widehat{\theta }_{T}&=& \mbox{arg} \ \min_{\theta \in \Theta }
\left\|-\int_{-\pi}^{\pi }p_{\omega }^{(T)}
\ln\left(\Upsilon_{\omega ,\theta }\right)\mathcal{W}_{\omega} d\omega\right\|_{\mathcal{L}(L^{2}(\mathbb{M}_{d},d\nu; \mathbb{C}))}. \label{mfetheta}
\end{eqnarray}

The following result provides weak--consistency of the minimum contrast parameter estimator $\mathcal{A}_{\widehat{\theta}_{T}}$ of the log--memory operator.
\begin{theorem} \label{th2} 
Let $\{\widetilde{X}_{t},\ t\in \mathbb{Z}\}$  be a  Gaussian  $L^{2}(\mathbb{M}_{d},d\nu)$--valued sequence satisfying
 \textbf{Conditions C0--C1}, as well as the conditions imposed through equations (\ref{eqdeftildew})--(\ref{identityintegralop})
 with $\gamma >1.$
 Assume that the slowly varying functions  $M_{n}(\omega ),\ \omega \in [-\pi,\pi],$ $n \in \mathbb{N}_{0}$    in equation (\ref{eqsmc1}) are   such that,
 for any $\xi>0,$ $
\lim_{\omega \to 0}\left[\sup_{n\in \mathbb{N}_{0}}\left|\frac{M_{n}\left(\omega /\xi \right)}{M_{n}\left(\omega \right)}
-1\right|\right]=0.$ Then,  the following limit holds:
\begin{eqnarray}
&& E\left\|\int_{-\pi}^{\pi}\left[p_{\omega
}^{(T)}-\mathcal{F}_{\omega ,\theta_{0}}\right]\mathcal{W}_{\omega
,\theta }d\omega \right\|_{\mathcal{S}(L^{2}(\mathbb{M}_{d},d\nu; \mathbb{C}))} \to
0,\quad T\to \infty, \label{eqtlimits1}
\end{eqnarray}
\noindent
where, for $(\omega , \theta)\in [-\pi,\pi]\backslash \{0\}\times \Theta ,$
$
\mathcal{W}_{\omega,\theta }= \ln\left(\Upsilon_{\omega ,\theta}\right)\mathcal{W}_{\omega}.$
 The minimum contrast estimator (\ref{mfetheta})  then satisfies
$ \widehat{\theta }_{T}\to_{P} \theta_{0},$ as  $T\to \infty ,$
 where $\to_{P}$ denotes   convergence   in probability.
\end{theorem} 

\begin{proof}     The proof follows as in    Theorem 2 in \cite{RuizMedina2022}.
\end{proof}

\medskip

\noindent \emph{Global analysis.}  Condition (\ref{hsc})   is a key condition in our approach,   leading to equation (\ref{klexpc2c}) introducing  the Hilbert--Schmidt operator setting, usually considered in functional time series analysis. Note that, under this setting,  LRD still can be displayed (see equation (\ref{sqishsdo})). Furthermore, equation (\ref{klexpc2c}) allows to apply the spectral analysis of LRD functional time series introduced in \cite{RuizMedina2022}  for inference on manifold cross--time RFs. Specifically,    equation  (\ref{hsc}) ensures Lemma 1 in \cite{RuizMedina2022} holds under alternative conditions  for this RF family. Then, asymptotic unbiasedness of the integrated periodogram also holds, beyond the Gaussian scenario under non structural assumptions (see Theorem \ref{th1}).  Furthermore,   Section \ref{mce}, and, in particular,  \textbf{Condition C1}, provides the semiparametric functional  spectral scenario introduced in \cite{RuizMedina2022} to be applied for   minimum contrast estimation, when  a functional sample of a manifold cross--time RF can be observed. Theorem \ref {th2} ensures weak--consistency under a Gaussian scenario as given in  \cite{RuizMedina2022}.
\section{SRD--LRD estimation in the spectral domain}\label{sec:4}

This section has a double, theoretical and practical, motivation. Specifically, on the one hand a wider family of spatiotemporal RF models is analyzed displaying SRD and LRD at different manifold scales, in the spirit of the LRD  framework introduced in \cite{LiRobinsonShang19}, but extended to the multifractional context.  On the other hand, this manifold--scale varying memory behavior can be observed in some stochastic  fractional or multifractional pseudodifferential  in time evolution equations, defined from a   local spatial differential operator, where the decay velocity of the temporal correlation function is accelerated by the decay of the spatial pure point spectrum, modifying the temporal  LRD level  of the model at the smallest spatial scales.   The reverse situation  can be observed  in  processes defined by  evolution equations given by a local differential operator in time, and a   fractional or multifractional pseudodifferential  operator in space.
Thus, the solution to these models  displays a spatial fractal behavior, reflected in a slow decay of its spatial pure point spectrum, slowing down
the decay velocity of the temporal correlation function,
leading to a stronger dependence at small spatial scales. Both behaviors can be observed within  the model family introduced in \cite{Anh2018}; \cite{AnhLeonenkoa}; \cite{AnhLeonenkob}.

 Let us now   consider
 in \textbf{Condition C1},  $l_{\alpha }(\theta )=0,$ and $\alpha (n,\theta )=0,$ for  $n\in \mathcal{D}_{\mbox{\small SRD}}\subset \mathbb{N}_{0},$   meaning that the process projected  into the eigenspaces $H_{n},$   $n\in \mathcal{D}_{\mbox{\small SRD}}\subset \mathbb{N}_{0},$ of the spherical  Laplace Beltrami operator displays SRD. While LRD is observed at the remaining eigenspaces.
  Without loss of generality let $\mathcal{D}_{\mbox{\small SRD}}=
 \{ 0,\dots n_{0}\},$  for certain $n_{0}\geq 1.$ (The reverse situation where the eigenvalues of the LRD operator vanish at large $n$ has been analyzed in the simulation study in Section \ref{sim2}).
  The projected   SRD process then admits  the following expansion (see Theorem 1 in the Supplementary Material):
\begin{eqnarray}&&\widetilde{X}_{t}^{(n_{0})}(\mathbf{x})=
\sum_{n=0}^{ n_{0}}\sum_{j=1}^{\delta (n,d)} V_{n,j}(t) S_{n,j}^{d}(\mathbf{x}),\ \mathbf{x}\in \mathbb{M}_{d},\ t=1,\dots,T.\nonumber
 \end{eqnarray}
\noindent
 For $\omega\in [-\pi ,\pi],$ the   fDFT  projected  into $\oplus_{n=0}^{n_{0}}H_{n}$
  is  expressed as
\begin{eqnarray}&&\widetilde{X}^{(T, n_{0} )}_{\omega}(\mathbf{x})
= \sum_{n=0}^{ n_{0}}\sum_{j=1}^{\delta (n,d)} \left[\frac{1}{\sqrt{2\pi
T}}\sum_{t=1}^{T}  V_{n,j}(t) \exp\left(-i\omega t\right)\right]S_{n,j}^{d}(\mathbf{x}),\ \mathbf{x}\in \mathbb{M}_{d}.\nonumber
\end{eqnarray}
\noindent   The corresponding  projected
  periodogram operator
$  p_{\omega }^{(T, n_{0})}=\widetilde{X}_{\omega }^{(T,n_{0})}\otimes
\overline{\widetilde{X}_{\omega}^{(T, n_{0})}}$ then
involves the tensorial product of the eigenfunctions of the Laplace Beltrami operator up to order $n_{0}.$ For $\omega \in [-\pi,\pi],$ the   kernel estimator $$\widehat{f}_{\omega }^{(T, n_{0})}(\mathbf{x},\mathbf{y})=\left[\frac{2\pi}{T}\right]\sum_{t\in [1,T-1]}  W^{(T)}\left(\omega - \frac{2\pi t}{T} \right) p_{2\pi t/T}^{(T, n_{0})}(\mathbf{x},\mathbf{y}),\ \mathbf{x},\mathbf{y}\in \mathbb{M}_{d},$$
\noindent of $\mathcal{F}_{\omega } $ projected into $\oplus_{n=0}^{n_{0}}H_{n},$ based on the weighted periodogram operator, is  computed. Here,
 $
W^{(T)}(x) = \sum_{j\in \mathbb{Z}}\frac{1}{B_{T}} W\left(\frac{x + 2\pi j}{B_{T}}\right),$
 with $B_{T}$ being the positive  bandwidth parameter. Function  $W$ on $\mathbb{R}$ is  real, and  such that
$W$ is positive, even, and bounded in variation, with $W(x) =0$ if  $ |x|\geq 1,$
$\int_{\mathbb{R}} \left|W(x)\right|^{2}dx <\infty,$ and  $\int_{\mathbb{R}} W(x)dx =1$  (see \cite{Panaretos13}).

Under
\textbf{Condition C1}, the minimum contrast estimators
 $\widehat{\alpha}_{T}(n, \theta_{0} )=$\linebreak $\alpha (n,\widehat{\theta}_{T}),$ $n>n_{0},$ are computed  as given in equations (\ref{eqdeftildew})--(\ref{mfetheta}). Hence, the mixed SRD--LRD kernel estimator of the spectral density operator is  given by
\begin{eqnarray}&&\widehat{f}_{\omega}^{(T)}(\mathbf{x},\mathbf{y})=\widehat{f}_{\omega  }^{(T,n_{0})}(\mathbf{x},\mathbf{y})+\widehat{f}_{\omega  }^{(T,(n_{0}+1,\infty))}(\mathbf{x},\mathbf{y},\widehat{\theta}_{T})\nonumber\\
&&=\left[\frac{2\pi}{T}\right]\sum_{t\in [1,T-1]}  W^{(T)}\left(\omega - \frac{2\pi t}{T} \right) p_{2\pi t/T}^{(T, n_{0})}(\mathbf{x},\mathbf{y})\nonumber\\
&&+
\sum_{n=n_{0}+1}^{\infty}\frac{B_{n}^{\eta }(0)M_{n}(\omega )}{\left[4(\sin(\omega /2))^{2}\right]^{\alpha (n,\widehat{\theta}_{T})/2}}\sum_{j=1}^{\delta (n,d)}S_{n,j}^{d}(\mathbf{x})S_{n,j}^{d}(\mathbf{y}).
\label{mixest}
\end{eqnarray}

\setcounter{section}{4} \setcounter{equation}{0} 
\section{Simulation study}\label{sec:6}
This section illustrates the results derived in the context of
 multifractionally integrated SPHARMA(p,q) processes (see, e.g., Example 1 in Section 3.3 in \cite{RuizMedina2022}, and \cite{LiRobinsonShang19} in the particular case  of fractionally integrated functional autoregressive moving averages processes). See also \cite{CaponeraEJS} and \cite{CaponeraMarinucci} for the case of SPHARMA(p,q) processes.

 Considering the unit sphere $\mathbb{S}_{2}$ in $\mathbb{R}^{3},$ and hence,  the separable Hilbert space $H=L^{2}(\mathbb{S}_{2},
 d\nu),$   a   multifractionally integrated SPHARMA(p,q) process \linebreak $\left\{\widetilde{X}_{t},\ t\in \mathbb{Z}\right\}$ is defined by  the following  state space equation:
\begin{eqnarray}&&
(\mathcal{I}_{L^{2}(\mathbb{S}_{2},d\nu)}-B)^{\mathcal{A}_{\theta }/2}(\boldsymbol{\Phi}_{p}(B)\widetilde{X}_{t})(\mathbf{x})=
\varepsilon_{t}(\mathbf{x})
+
(\boldsymbol{\Psi}_{q}(B)\varepsilon_{t})(\mathbf{x}),\quad \mathbf{x}\in \mathbb{S}_{2},\ t\in \mathbb{Z}.\nonumber\\
\label{ex2modeldef}
\end{eqnarray}
\noindent
\textbf{Condition C1} holds,  since
this model constitutes a particular case  $H=L^{2}(\mathbb{S}_{2},
 d\nu)$  of the more general formulation
 given  in Section 3.3 in \cite{RuizMedina2022} in the  functional time series framework.
Then, as before, $\mathcal{A}_{\theta } $ is the LRD operator satisfying $l_{\alpha }(\theta ), L_{\alpha }(\theta )\in (0,1/2).$  In particular,
equation  (\ref{sqishsdo}) holds.
Here, operator $(\mathcal{I}_{L^{2}(\mathbb{S}_{2},d\nu)}-B)^{\mathcal{A}_{\theta }/2}$   is interpreted
as  in \cite{Charac14};  \cite{Rackauskasv2}, and
 $B$ is a difference operator satisfying
$
E\|B^{j}\widetilde{X}_{t}-\widetilde{X}_{t-j}\|_{H}^{2}=0,$ for $t,j\in \mathbb{Z}.$

 Here, $\boldsymbol{\Phi}_{p}(B)=1-\sum_{k=1}^{p}\Phi_{k}B^{k},$ and $\boldsymbol{\Psi}_{q}(B)=\sum_{l=1}^{q}\Psi_{l}B^{l},$
where operators $\Phi_{k},$ $k=1,\dots,p,$ and $\Psi_{l},$ $l=1,\dots,q,$ are assumed to be invariant positive self-adjoint bounded operators on $L^{2}(\mathbb{S}_{2},d\nu)$  admitting the representation:
\begin{eqnarray}
\Phi_{k}&=&\sum_{n\in \mathbb{N}_{0}}\lambda_{n}(\Phi_{k})\sum_{j=1}^{\delta (n,d)}S_{n,j}^{d}\otimes S_{n,j}^{d},\  k=1,\dots,p\nonumber\\
\Psi_{l}&=&\sum_{n\in \mathbb{N}_{0}}\lambda_{n}(\Psi_{l})\sum_{j=1}^{\delta (n,d)}S_{n,j}^{d}\otimes S_{n,j}^{d},\  l=1,\dots,q.
\label{de}
\end{eqnarray}
\noindent Here, $d=2,$ $\delta (n,2)=2n+1,$  $\mathbb{M}_{2}=\mathbb{S}_{2},$  $\omega_{2}=|\mathbb{S}_{2} |=4\pi.$
Also,  $\Phi_{p,n}(z)=1-\sum_{k=1}^{p}\lambda_{n}(\Phi_{k})z^{k}$ and $\Psi_{q,n}(z)=\sum_{l=1}^{q}\lambda_{n}(\Psi_{l})z^{l},$
 $n\in \mathbb{N}_{0},$
 have not common
roots, and  their roots are outside of the unit circle (see  Corollary 6.17 in \cite{Beran17}).
By  similar arguments, as justified in Sections 3.3 and 6 in   \cite{RuizMedina2022}, conditions (\ref{eqdeftildew})--(\ref{identityintegralop}) are satisfied by this model family.
Hence, the results given in Theorems \ref{th1} and \ref{th2}  hold.

  Theorem 1 and Section 1.1 of  the Supplementary Material have been applied in the simulation methodology  adopted in  the generations of some special cases within  the family of multifractionally integrated SPHARMA(p,q) processes.
   In particular, a spherical uniform pole $\mathbf{U}=\mathbf{u}_{0}$ (see Figure \ref{pole}) in the involved zonal functions (defined from  the Legendre polynomials $\left\{ P_{n},\ n\in \mathbb{N}_{0}\right\}$) has been  independently generated  of
      the $L^{2}(\mathbb{S}_{2},d\nu)$--valued Gaussian strong--white noise   innovation process   with variance $\sigma^{2}_{\varepsilon}=\sum_{n\in \mathbb{N}_{0}}[\sigma_{n}^{\varepsilon}]^{2}.$ Table 3 of the Supplementary Material provides  the specific parametric scenarios considered. Under such scenarios, Section 3 of the Supplementary Material  displays generations of some special cases of multifractionally integrated  SPHARMA(p,q) processes, for $p=1,3$ and $q=0,$ and  $p=1,3,$  and $q=1$ in equation (\ref{ex2modeldef}). The particular cases analyzed of   LRD operator eigenvalue sequences     can be found in Table 1 of the Supplementary Material.
 In the implementation of the minimum contrast estimation methodology, we consider $100$ candidate systems of  parametric    eigenvalues for the LRD operator (see Table 2 in the Supplementary Material).
   These candidate sets are displayed in Figure \ref{figlrdeigcand} jointly with the true eigenvalue sequence.
\begin{figure}[!htb]
\begin{center}
\includegraphics[height=5cm, width=4.5cm]{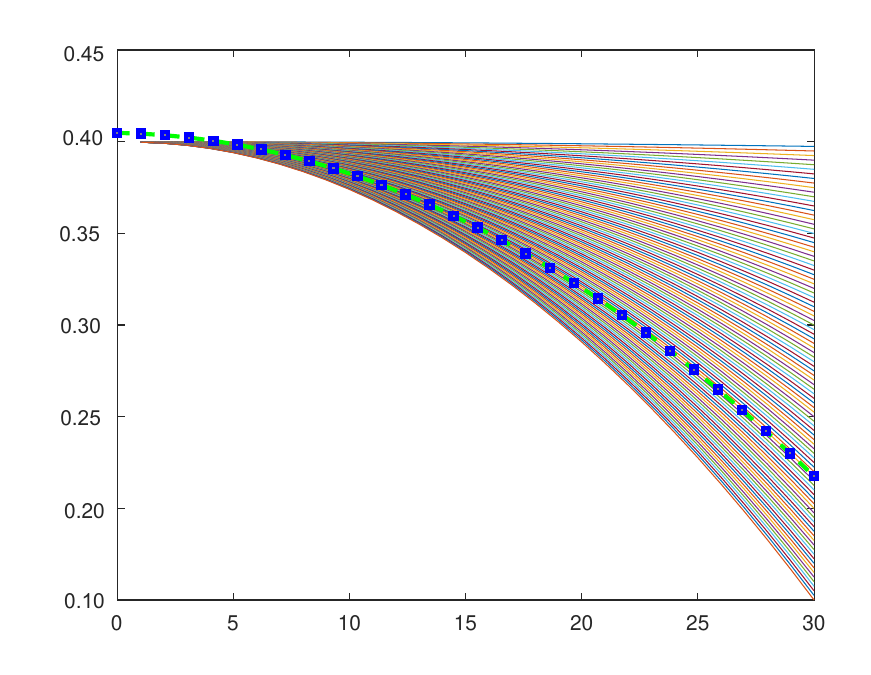}
\includegraphics[height=5cm, width=4.5cm]{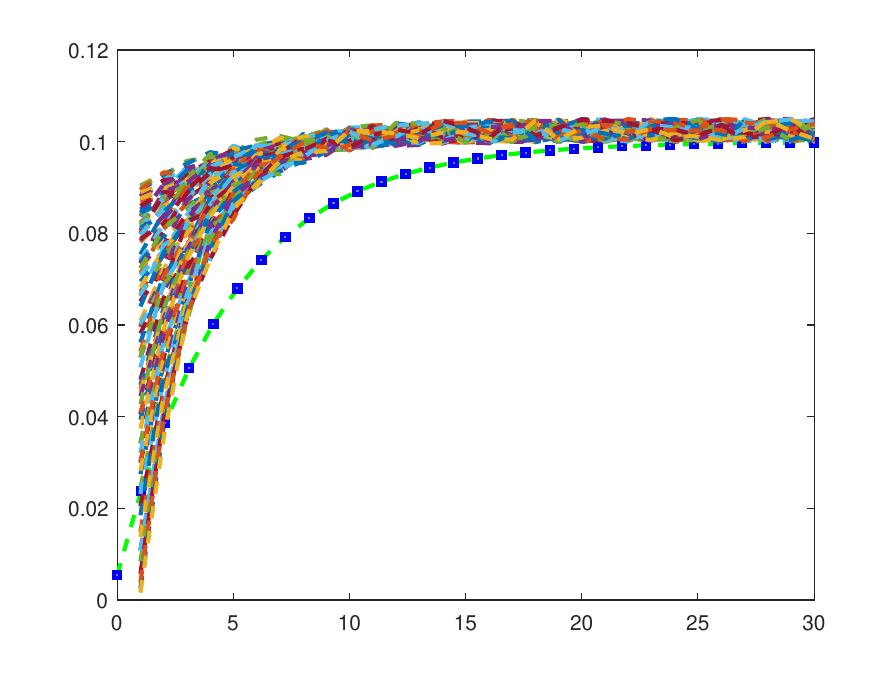}
\caption{ \scriptsize{The    first $30$ eigenvalues  $\alpha (n,\theta_{0}),$   $n=1,\dots,30$  of the  LRD operator $\mathcal{A }_{\theta_{0} }$ (dotted blue--green line), and the set of $100$ parametric candidates $\alpha (n,\theta_{i} ),$  $n=1,\dots,30,$  $i=1,\dots,100,$   for decreasing  LRD operator eigenvalues  (left--hand side). The  first $30$ eigenvalues   $\alpha (n,\vartheta_{0}),$  $n=1,\dots,30$   (dotted blue--green line)   of the  LRD operator $\mathcal{A }_{\vartheta_{0} },$  and the set of $100$ parametric candidates $\alpha (n,\vartheta_{i} ),$  $n=1,\dots,30,$  $i=1,\dots,100,$  for increasing LRD operator eigenvalues  (right--hand side)}}\label{figlrdeigcand}
\end{center}
\end{figure}

        The functional spherical values at different times of the generated  multifractionally integrated    SPHAR(3) process are displayed
      for $M=10$  under decreasing  eigenvalue sequence of the LRD operator   in Figure  \ref{sphar3}. As commented,  the remaining special cases are plotted in Section 3 of the Supplementary Material.  Additionally, generations under LRD operator with eigenvalues vanishing for $n\geq n_{0}=16,$ i.e., the   LRD--SRD case, is showed in Section 3.1 of the Supplementary Material for multifractionally integrated SPHARMA(1,1) process.

It can be observed in all generations,  displayed on a spherical  angular neighborhood of the selected pole (see Figure \ref{pole}), that  persistent in time of the local spherical behavior   is stronger   when a positive bounded non--decreasing sequence of eigenvalues of  LRD operator is considered, since  the highest positive eigenvalues are displayed at high discrete Legendre frequencies.  For decreasing eigenvalue sequence  of the  LRD operator,   the opposite LRD effect is observed through  spherical scales.   One can also observe a  symmetry in the evolution in time of the spherical patterns, under both, decreasing  and increasing  LRD operator eigenvalue sequences,  when spatial  spherical SRD is observed. An increasing level of   local linear correlation in space at intermediate times  is displayed  when autoregression of order $3$ is considered.
\begin{figure}[H]
\centering\includegraphics[scale=0.4]{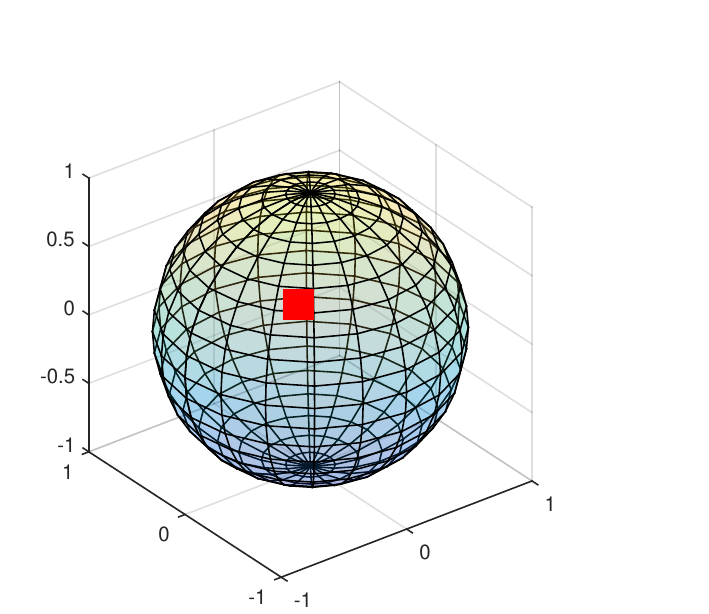}
\caption{ \scriptsize{The selected pole in the zonal functions}}\label{pole}
\end{figure}
\noindent
\begin{figure}[H]
\centering
\includegraphics[scale=0.7]{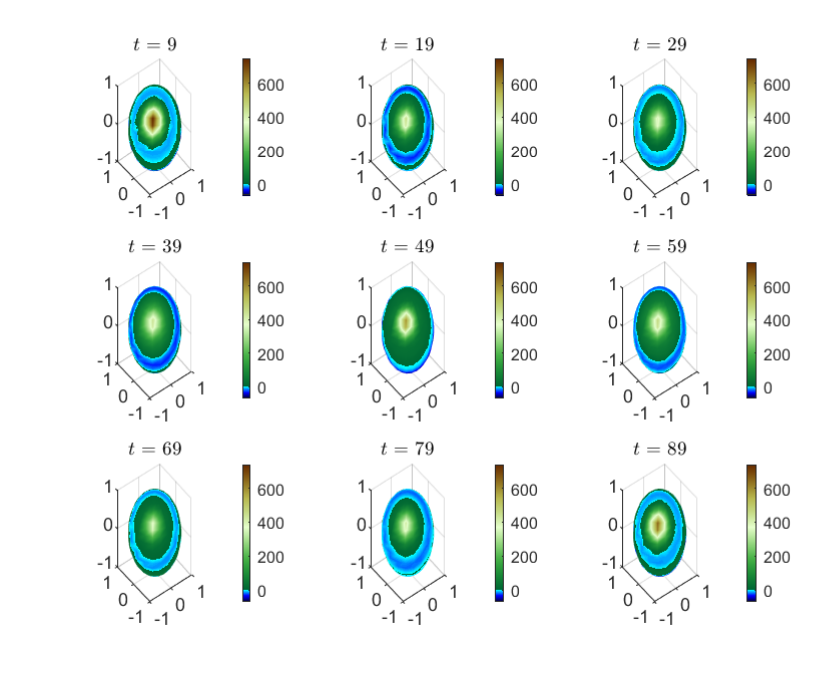}
\caption{ \scriptsize{Generations  of multifractionally integrated  SPHAR(3) process $\widetilde{X}$ at  times $t=9,19,29, 39, 49, 59, 69, 79, 89,$  projected into the direct sum  $\oplus_{n=1}^{10}H_{n}$ of eigenspaces of the spherical Laplace Beltrami operator on $L^{2}(\mathbb{S}_{2}, d\nu)$ under decreasing eigenvalue sequence of the LRD operator}}
\label{sphar3}
\end{figure}

 \subsection{Minimum contrast estimation results}
\label{mcess}
We now display the minimum contrast estimation results referred to the multifractionally integrated  SPHAR(3) process, when   LRD operator  $\mathcal{A}_{\theta }$ has  decreasing sequence of eigenvalues as plotted at the left--hand side of Figure \ref{figlrdeigcand}, under a truncation order $M=30.$   See  Section 4 of the  Supplementary Material, for the remaining cases of multifractionally integrated  SPHAR(1),  SPHAR(3),  SPHARMA (1,1), SPHARMA(3,1) processes.
Specifically, Figure \ref{figsqrsdo} displays, for $i=1,\dots,100$ frequency nodes, operator $|\mathcal{M}_{\omega _{i}} |^{1/2}$    (left--hand side), and  $\left|4(\sin(\omega_{i} /2))^{2}\right|^{-\mathcal{A}_{\theta}/4} |\mathcal{M}_{\omega _{i}}|^{1/2}$ (right--hand side), projected into $H_{n},$  $n=1,\dots, 30.$  These factors have been computed in  the implementation of the minimum contrast estimation methodology in the functional spectral domain.
Figure \ref{SDO} provides the projected  spectral density operator kernel  at temporal Fourier frequencies $\omega = -\pi+0.0628(10)i,$ $i=1,3,7,10$ in the interval $[-\pi,\pi].$ Its empirical counterpart is given in Figure \ref{figsdp} where
the modulus of the projected  fDFT, and  tapered periodogram operator kernel of the generated data at temporal Fourier  frequency zero are displayed. The last one  over a grid of $30\times 30$ Legendre frequencies. The projected empirical contrast operator $\mathbf{U}_{T,\theta }$ in equation (\ref{eco}) is then computed. Its minimization is performed in the bounded operator norm.
\begin{figure}[H]
\begin{center}
\includegraphics[height=0.15\textheight, width=0.4\textwidth]{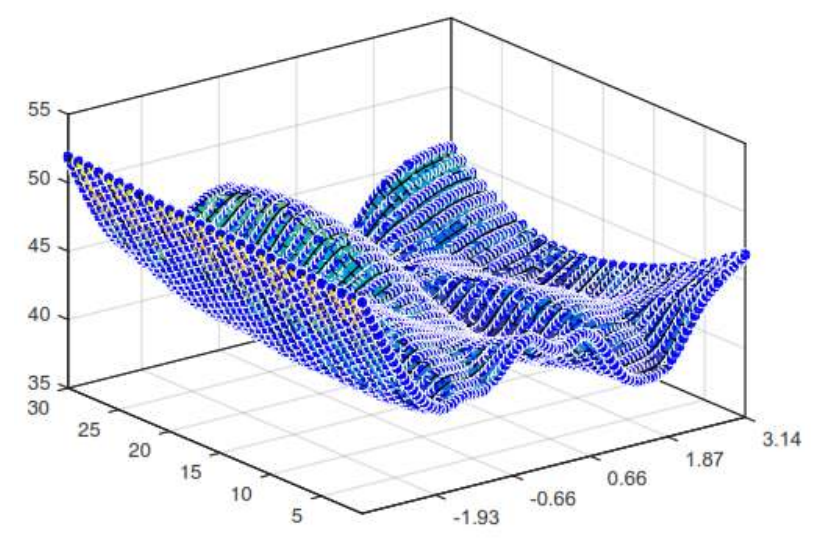}
\includegraphics[height=0.15\textheight, width=0.4\textwidth]{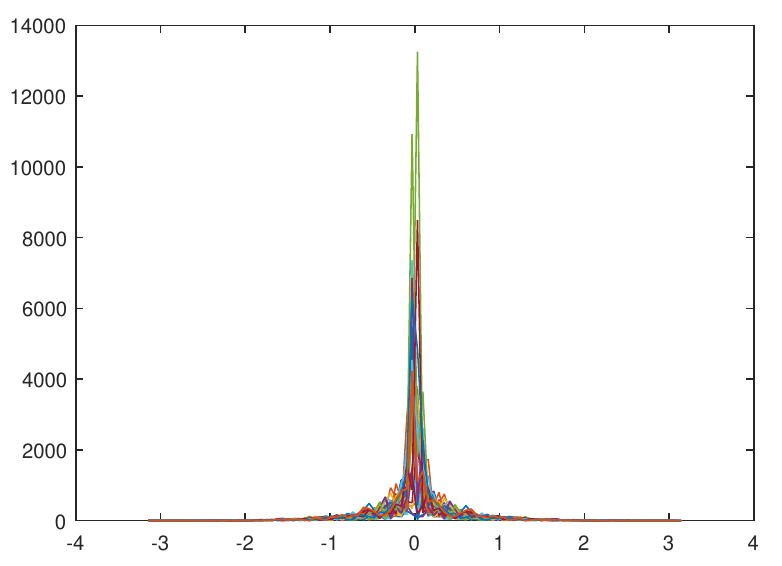}
\caption{ \scriptsize{The square root (s.r.) of the  modulus of the projected regular factor of the spectral density operator   (left--hand side), and
the product of this factor with the  s.r.  of the modulus of the  singular  factor   of the spectral density operator  (right--hand side) at Legendre frequencies $n=1,\dots ,30$}}\label{figsqrsdo}
\end{center}
\end{figure}
\begin{figure}[H]
\begin{center}
\includegraphics[height=0.25\textheight, width=0.45\textwidth]{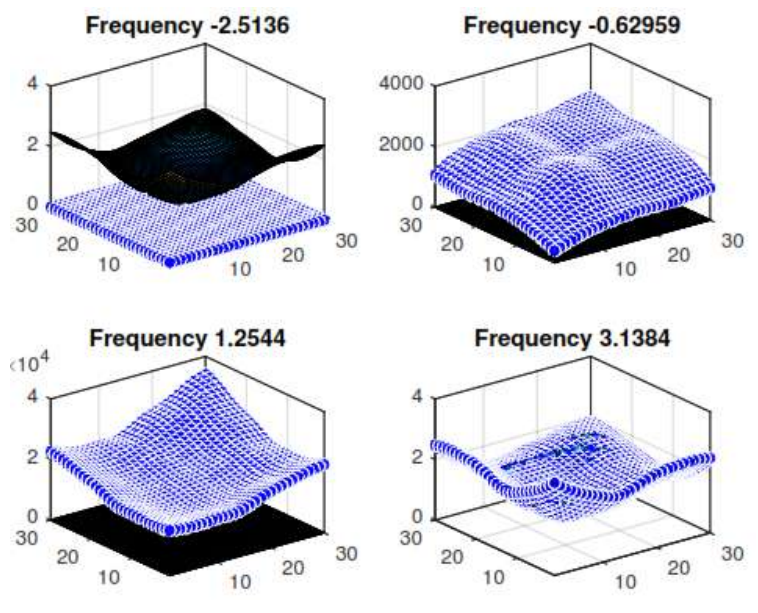}
\caption{ \scriptsize{Kernel of spectral density operator    $\mathcal{F}_{\omega ,\theta_{0}}$ projected into $H_{n}\otimes H_{n},$  $n=1,\dots,30,$   for temporal Fourier frequencies $\omega = -\pi+0.0628(10)i,$ $i=1,3,7,10$}}\label{SDO}
\end{center}
\end{figure}
\begin{figure}[H]
\begin{center}
\includegraphics[height=0.15\textheight, width=0.3\textwidth]{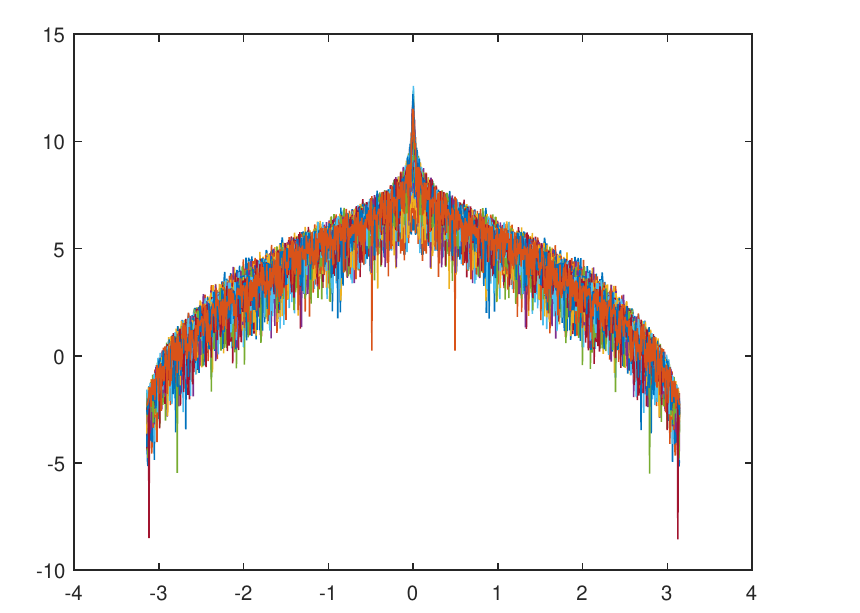}
\includegraphics[height=0.15\textheight, width=0.3\textwidth]{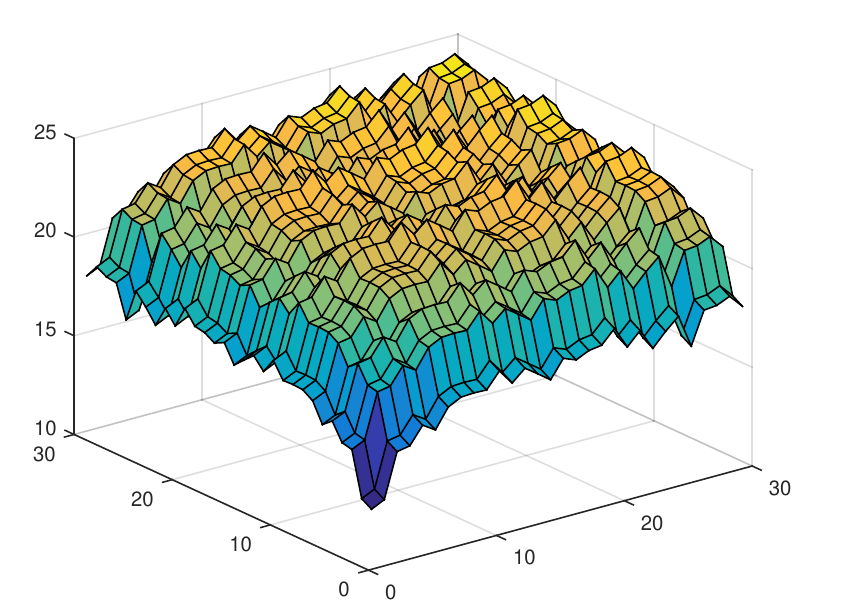}
\caption{ \scriptsize{The modulus of the projected  fDFT of the data, $30$ spectral curves corresponding to $30$ discrete Legendre  frequencies (left--hand side).  Tapered periodogram kernel at zero   temporal Fourier frequency over a $30\times 30$  grid of discrete Legendre  frequencies (right--hand side).
Both, fDFT and tapered periodogram kernel, are displayed at the logarithm scale }}\label{figsdp}
\end{center}
\end{figure}

Under the LRD operator scenario plotted   at the left--hand side of  Figure \ref{figlrdeigcand} (see also Table 1 of the Supplementary Material),    Figure \ref{fhistT1000sphar3}
displays the  histograms of the temporal mean of the  empirical  absolute errors from $R=100, 2000, 5000$ independent generations of
 a functional
sample of size $T=1000$ of  multifractionally integrated  SPHAR(3) process.  The empirical analysis performed for the remaining  cases  under decreasing LRD operator eigenvalue sequence, from   $R=100, 2000,5000$ independent generations of each one of the   functional samples of size  $T=50, 500,1000$ considered, are given in Sections 4.1, 4.3, 4.5, and 4.7 of the Supplementary Material (see also Sections 4.2, 4.4, 4.6 and 4.8 of the Supplementary Material for the same analysis under increasing  LRD operator eigenvalue sequence as plotted at the right--hand side of Figure \ref{figlrdeigcand}).
The results are displayed  for the eigenspaces $H_{n},$ $n=1,5, 10, 15, 20, 25, 30.$
The empirical probabilities
\begin{equation}
\widehat{\mathcal{P}}\left( \|f_{n,\theta_{0}}(\cdot )-f_{n,\widehat{\theta}_{T}}(\cdot )\|_{L^{1}([-\pi,\pi])}>\varepsilon_{i}\right),\ i=1,\dots, 100,
\label{ep}
\end{equation}
\noindent
are computed   under decreasing and increasing  LRD operator  eigenvalue  sequence, for the  multifractionally integrated SPHAR(3) model in Figure  \ref{femprobsph3}, considering
$R=100, 2000, 5000$ independent generations of each one of the  functional samples of size
 $T=50, 500, 1000.$
 In all scenarios displayed in Section 4 of the Supplemntary Material, the empirical  probabilities  (\ref{ep})  are computed  from projection into the eigenspaces $H_{n},$ $n=1,\dots, 30$   of the   Laplace Beltrami operator, and for a grid of $100$ thresholds in the interval $(0, 0.1).$

When the number of repetitions increases,  the tails of the  empirical distribution  of the temporal mean of the empirical  absolute errors become lighter, being  the  shape of the  empirical distribution closer to  a   Gaussian distribution. This effect associated with the increasing of the sample size  is more pronounced  at highest spatial resolution levels (i.e., at high discrete Legendre  frequencies). In particular, the empirical distribution of the temporal mean of the empirical absolute errors becomes more concentrated around its mean faster at  higher than at coarser  resolution levels.   A similar behavior is observed in  the  asymmetric empirical distribution of quadratic error temporal mean,  displaying a larger degree of dispersion,  due to the stronger effect of the functional  spectral singularity at zero frequency.
\begin{figure}[H]
     \centering
     \begin{subfigure}[t]{0.3\textwidth}
         \centering
         \includegraphics[height=0.135\textheight, width=\textwidth]{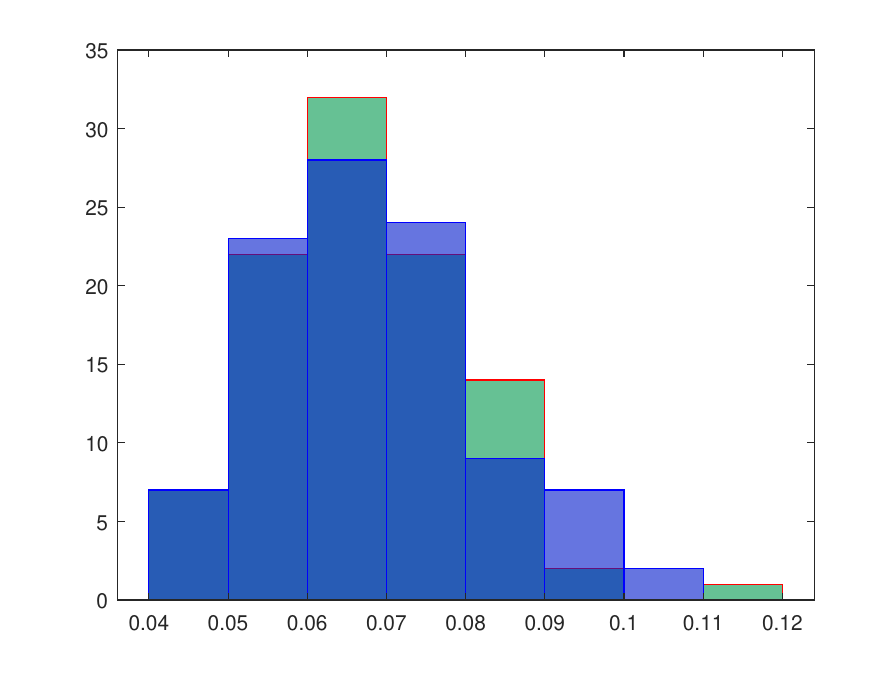}
         \caption{$R=100$, scales 1--5}
         \label{fhistR100es15T1000sphar3}
     \end{subfigure}
     \hfill
     \begin{subfigure}[t]{0.3\textwidth}
         \centering
         \includegraphics[height=0.135\textheight, width=\textwidth]{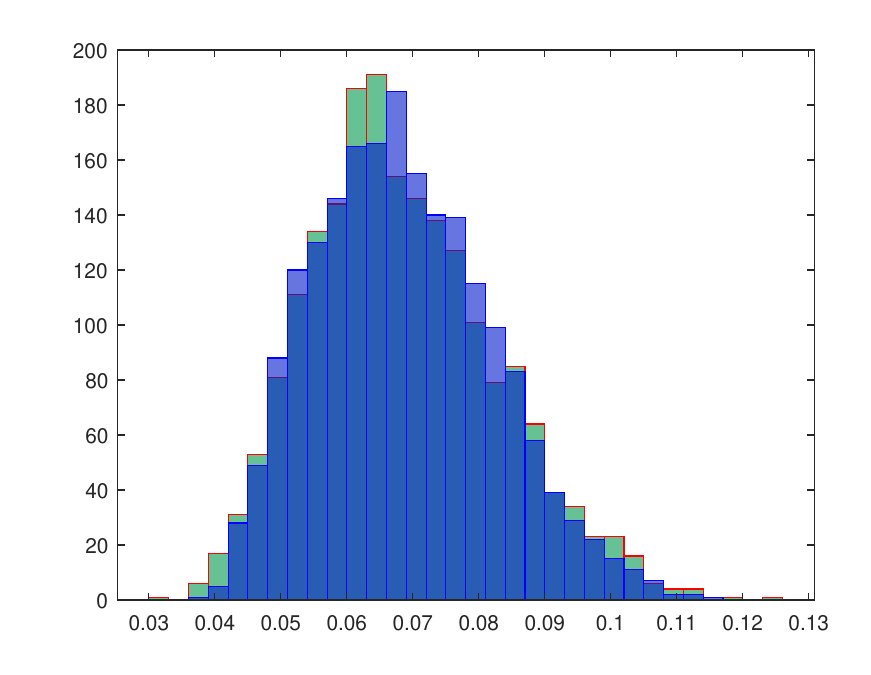}
         \caption{$R=2000$, scales 1--5}
         \label{fhistR2000es15T1000sphar3}
     \end{subfigure}
     \hfill
     \begin{subfigure}[t]{0.3\textwidth}
         \centering
         \includegraphics[height=0.135\textheight, width=\textwidth]{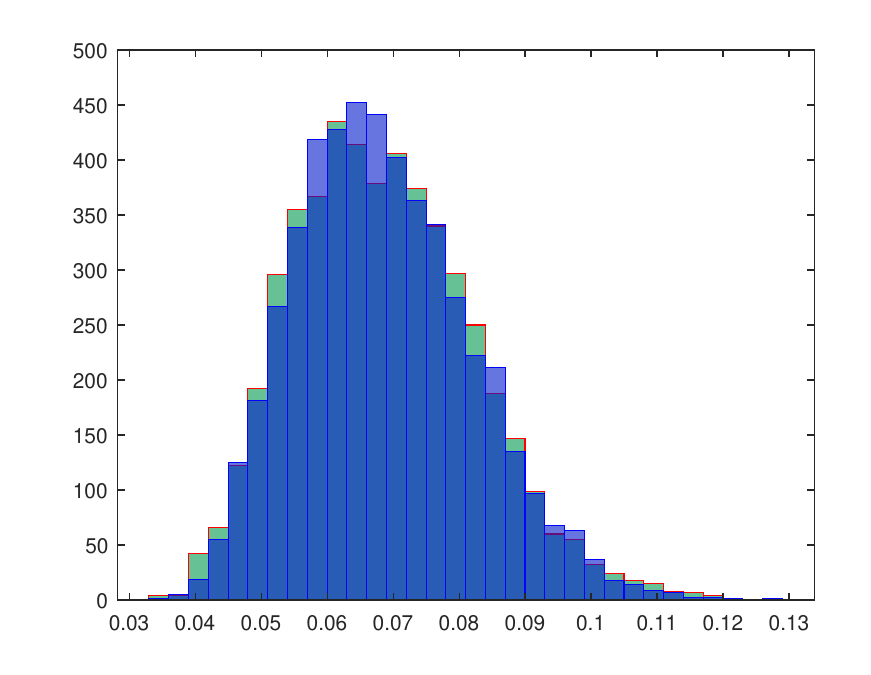}
         \caption{$R=5000$, scales 1--5}
         \label{fhistR5000es15T1000}
     \end{subfigure}
     \begin{subfigure}[c]{0.3\textwidth}
         \centering
         \includegraphics[height=0.135\textheight, width=\textwidth]{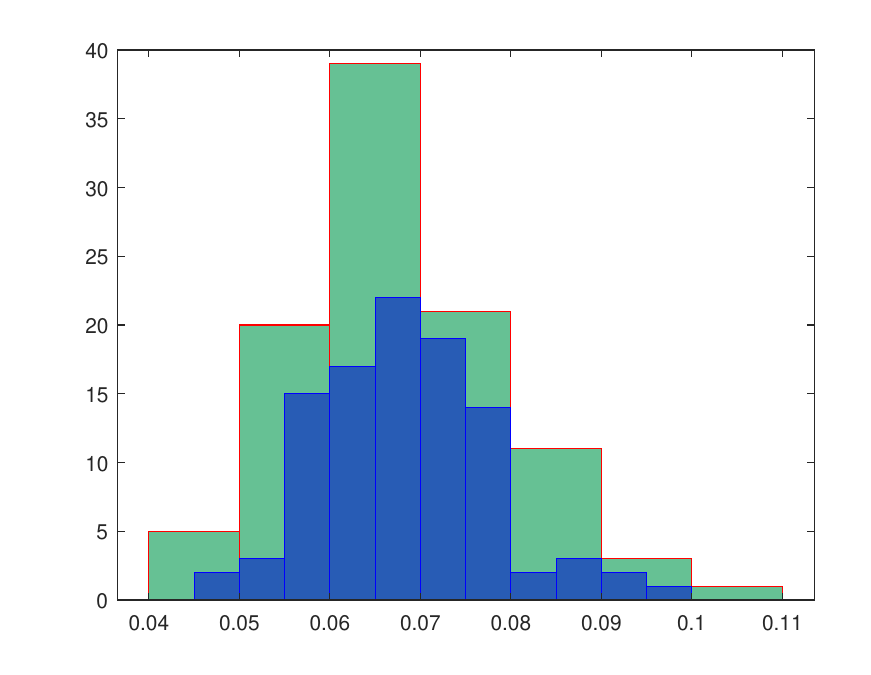}
         \caption{$R=100$, scales 10--15}
         \label{fhistR100es1015T1000sphar3}
     \end{subfigure}
     \hfill
     \begin{subfigure}[c]{0.3\textwidth}
         \centering
         \includegraphics[height=0.135\textheight, width=\textwidth]{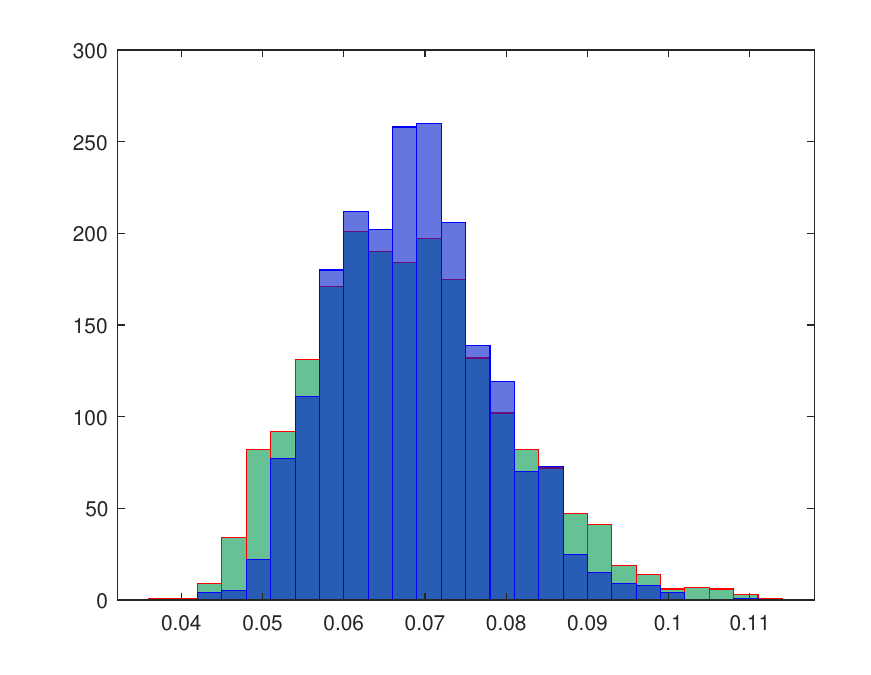}
         \caption{$R=2000$, scales 10--15}
         \label{fhistR2000es1015T1000sphar3}
     \end{subfigure}
     \hfill
     \begin{subfigure}[c]{0.3\textwidth}
         \centering
         \includegraphics[height=0.135\textheight, width=\textwidth]{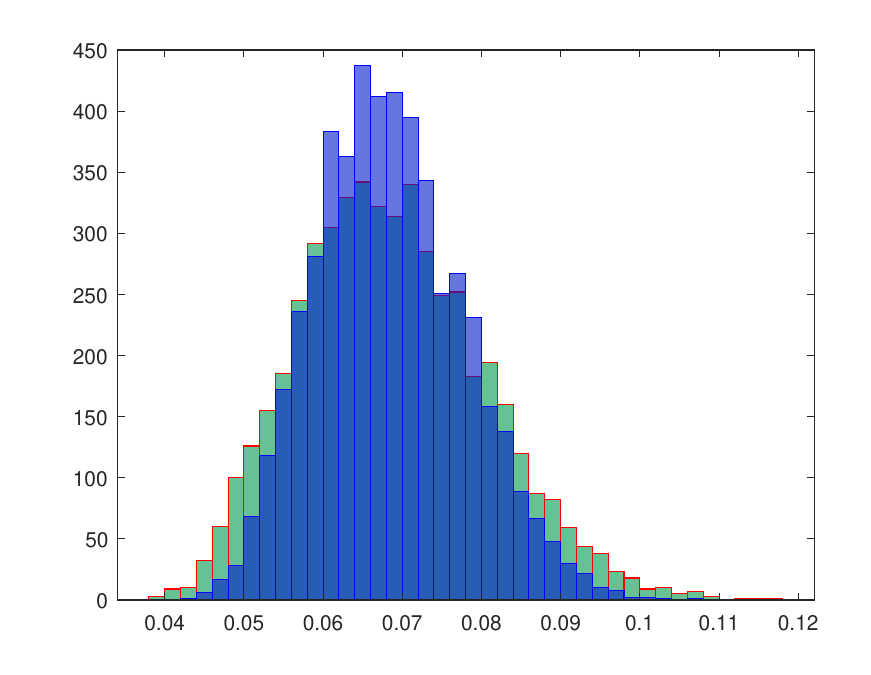}
         \caption{$R=5000$, scales 10--15}
         \label{fhistR5000es1015T1000sphar3}
     \end{subfigure}
        \begin{subfigure}[c]{0.3\textwidth}
         \centering
         \includegraphics[height=0.135\textheight, width=\textwidth]{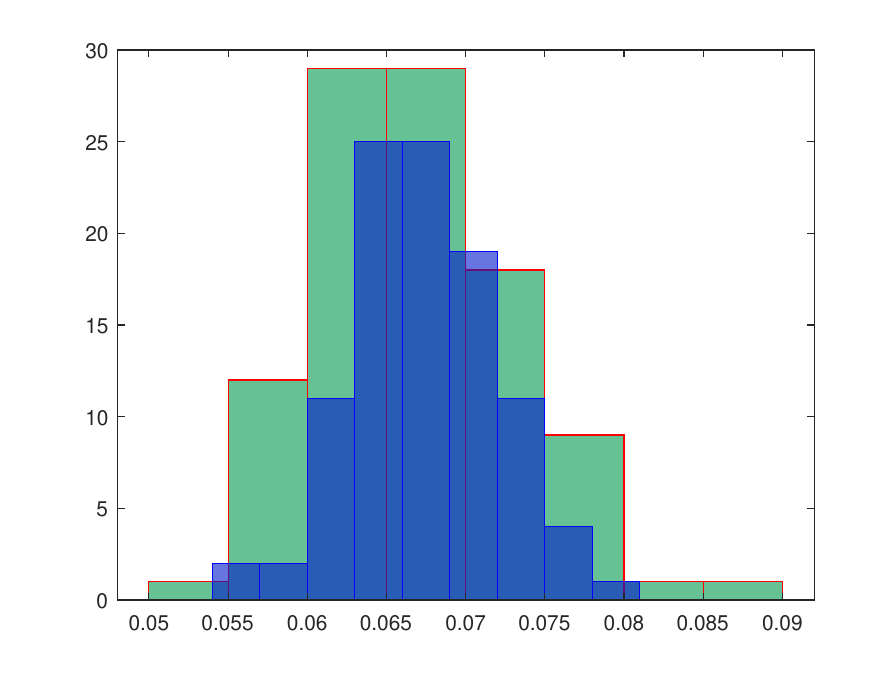}
         \caption{$R=100$, scales 20--25}
         \label{fhistR100es2025T1000sphar3}
     \end{subfigure}
     \hfill
     \begin{subfigure}[c]{0.3\textwidth}
         \centering
         \includegraphics[height=0.135\textheight, width=\textwidth]{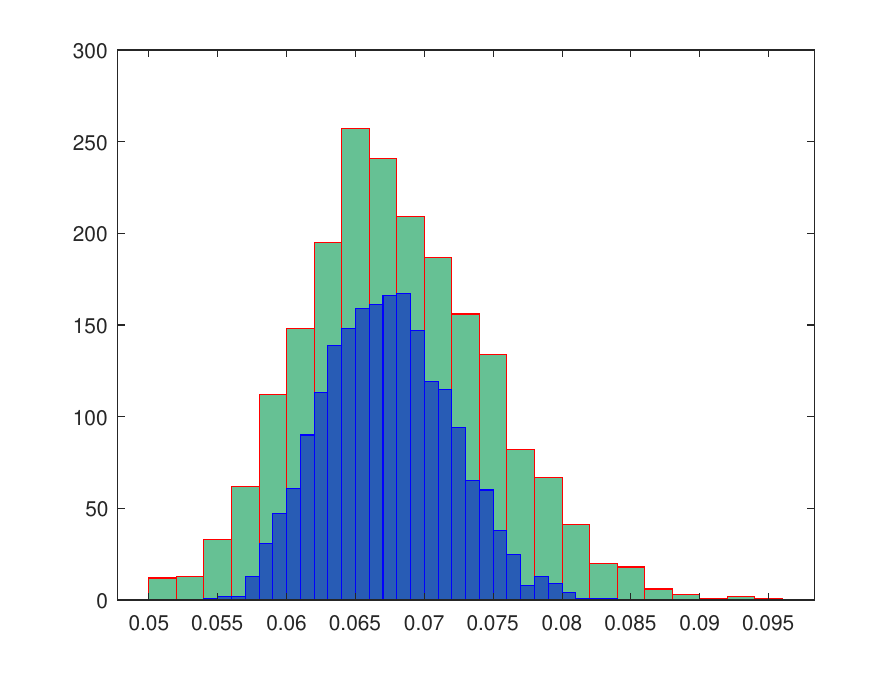}
         \caption{$R=2000$, scales 20--25}
         \label{fhistR2000es2025T1000sphar3}
     \end{subfigure}
     \hfill
     \begin{subfigure}[c]{0.3\textwidth}
         \centering
         \includegraphics[height=0.135\textheight, width=\textwidth]{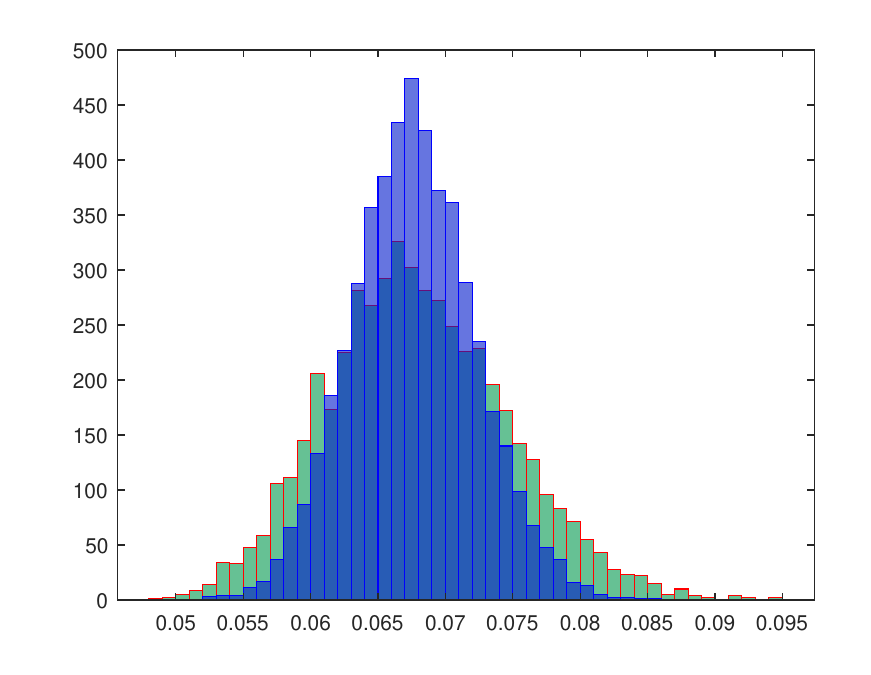}
         \caption{$R=5000$, scales 20--25}
         \label{fhistR5000es2025T1000sphar3}
     \end{subfigure}
      \hfill
     \begin{subfigure}[b]{0.3\textwidth}
         \centering
         \includegraphics[height=0.135\textheight, width=\textwidth]{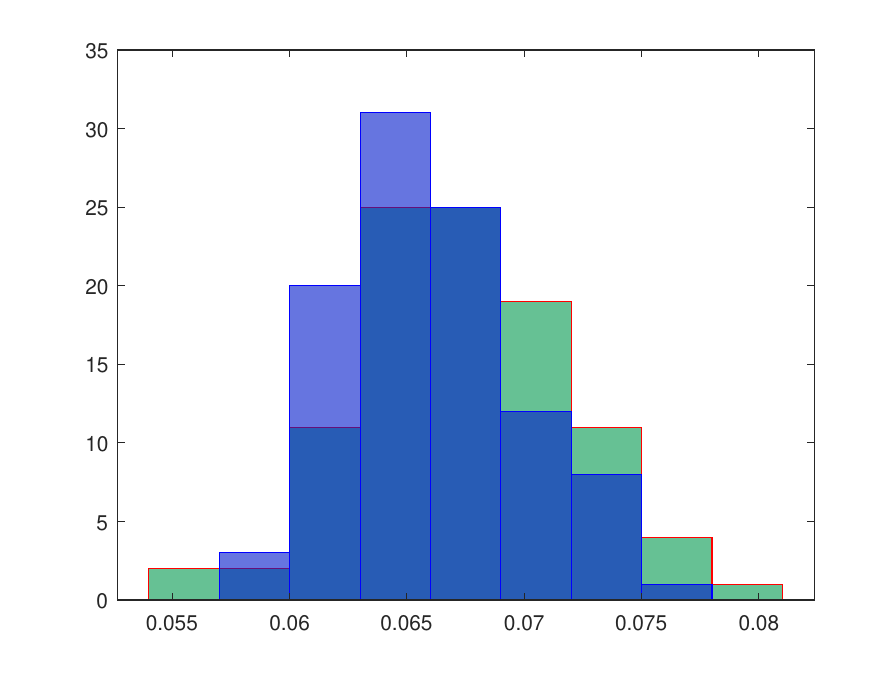}
         \caption{$R=100$, scales 25--30}
         \label{fhistR100es2530T1000sphar3}
     \end{subfigure}
           \hfill
     \begin{subfigure}[b]{0.3\textwidth}
         \centering
         \includegraphics[height=0.135\textheight, width=\textwidth]{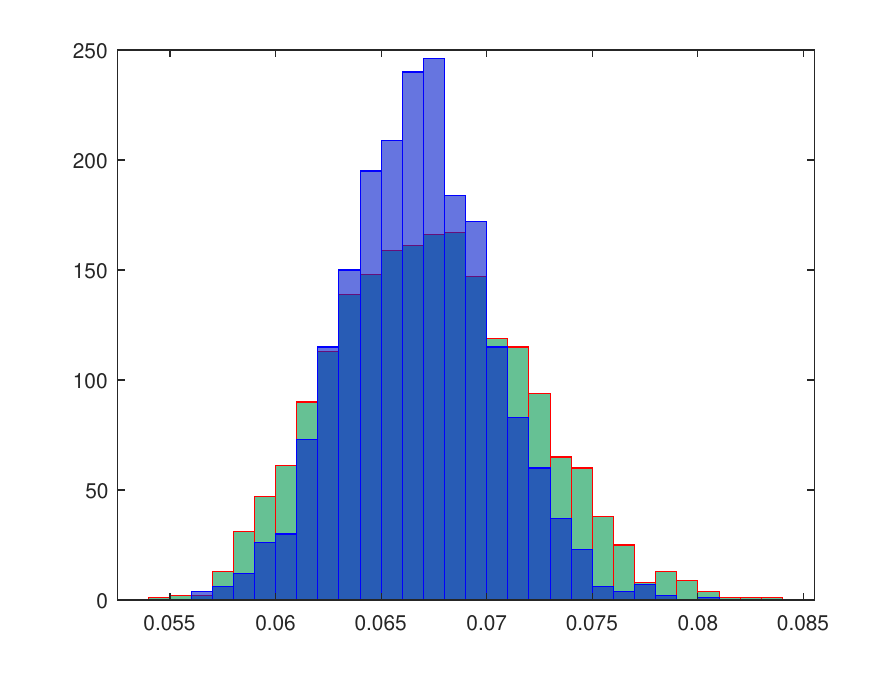}
         \caption{$R=2000$, scales 25--30}
         \label{fhistR2000es2530T1000sphar3}
     \end{subfigure}
           \hfill
     \begin{subfigure}[b]{0.3\textwidth}
         \centering
         \includegraphics[height=0.135\textheight, width=\textwidth]{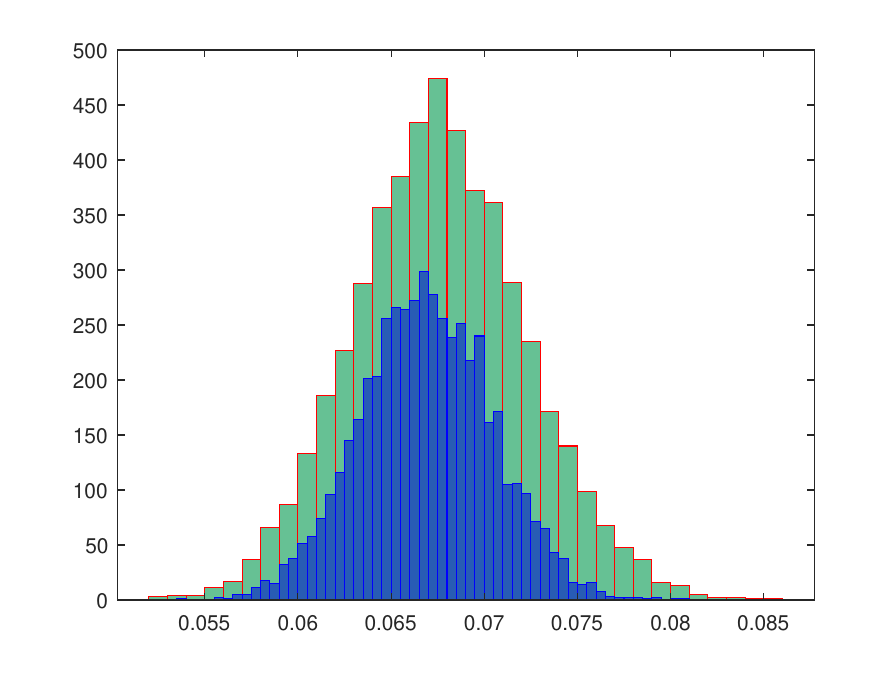}
         \caption{$R=5000$, scales 25--30}
         \label{fhistR5000es2530T1000sphar3}
     \end{subfigure}
        \caption{Histograms of the temporal mean of the empirical absolute errors from a functional  sample of  size $T=1000$ (LRD operator decreasing eigenvalue sequence)}
        \label{fhistT1000sphar3}
\end{figure}
Regarding the empirical probability analysis, one can observe that the increasing of $R$ has a strong effect by spherical scales on the gradually decay of the empirical probabilities to zero over the smallest threshold values  in the grid in the  interval $(0,0.1),$ while increasing  parameter $T$  enlarges the dark blue area, in the contour plots displayed, where  empirical probabilities become zero over  the largest threshold values in the grid analyzed.  Note that the opposite effects of parameter $R$ by spherical scales, in the increasing  and decreasing  LRD operator eigenvalue scenarios, are observed. All the results displayed are affected by a numerical integration error.
       \begin{figure}[H]
     \centering
     \begin{subfigure}[t]{0.3\textwidth}
         \centering
         \includegraphics[width=\textwidth]{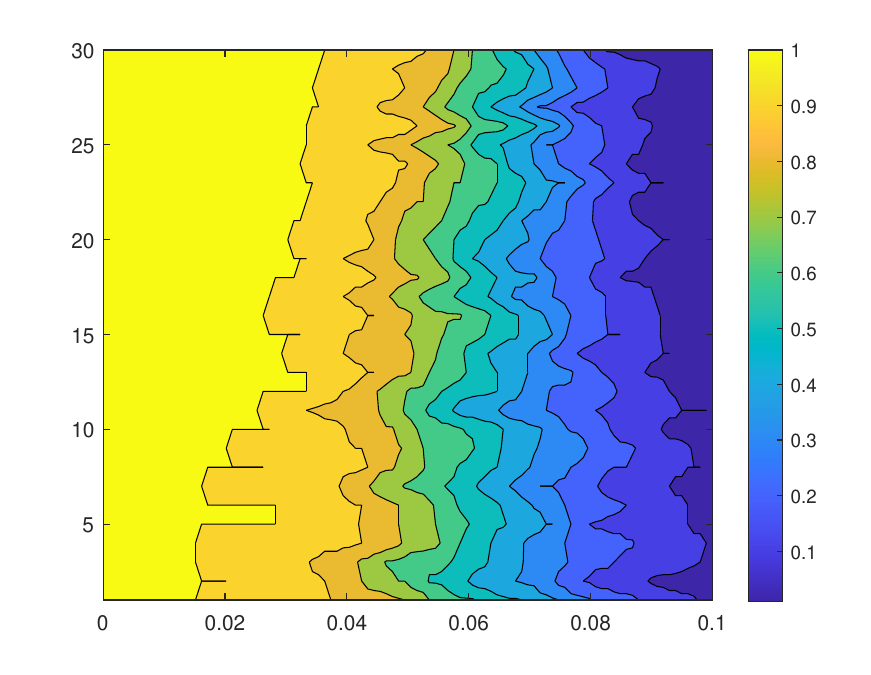}
         \caption{$T=50, R=100$}
         \label{femproT50R100sph3}
     \end{subfigure}
     \hfill
     \begin{subfigure}[t]{0.3\textwidth}
         \centering
         \includegraphics[width=\textwidth]{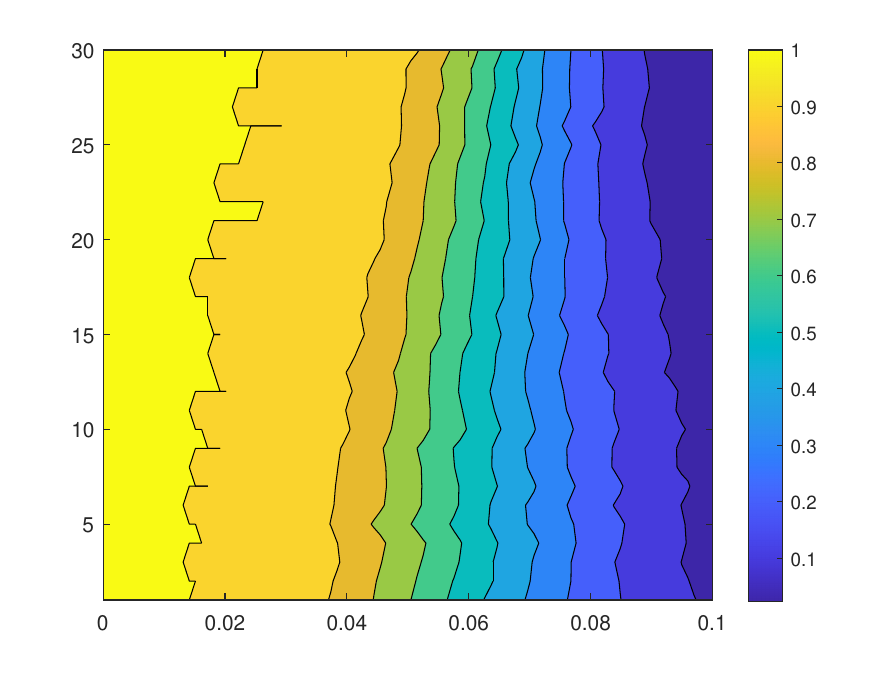}
         \caption{$T=50, R=2000$}
         \label{femproT50R2000}
     \end{subfigure}
     \hfill
     \begin{subfigure}[t]{0.3\textwidth}
         \centering
         \includegraphics[width=\textwidth]{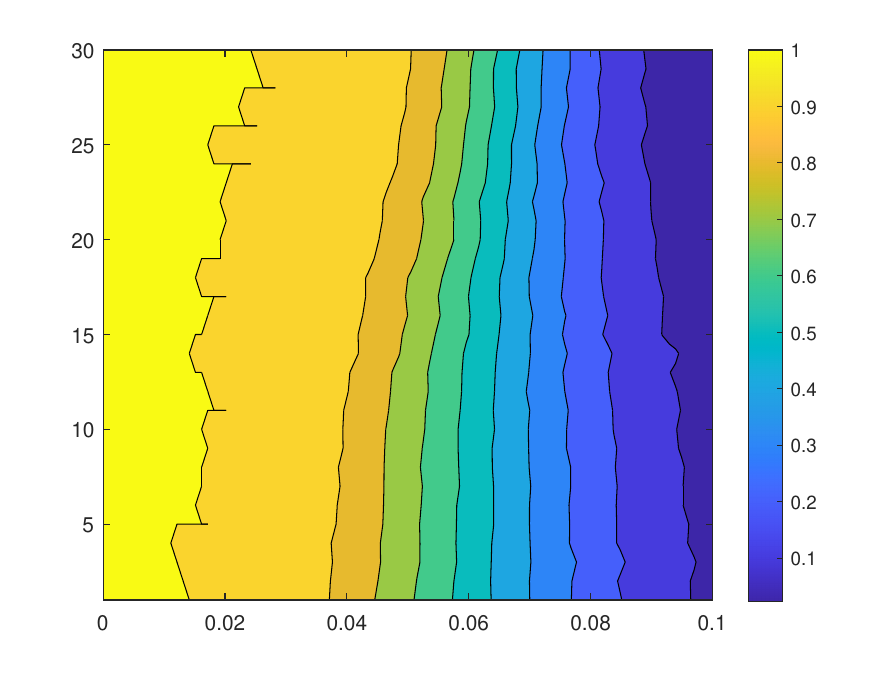}
         \caption{$T=50, R=5000$}
         \label{femproT50R5000}
     \end{subfigure}
     \begin{subfigure}[c]{0.3\textwidth}
         \centering
         \includegraphics[width=\textwidth]{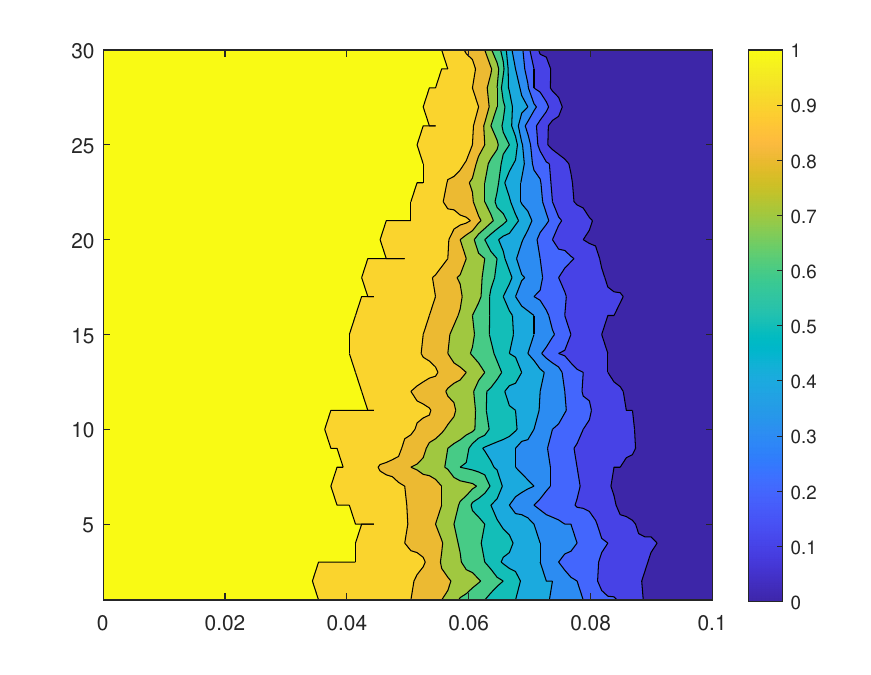}
         \caption{$T=500, R=100$}
         \label{femproT500R100}
     \end{subfigure}
     \hfill
     \begin{subfigure}[c]{0.3\textwidth}
         \centering
         \includegraphics[width=\textwidth]{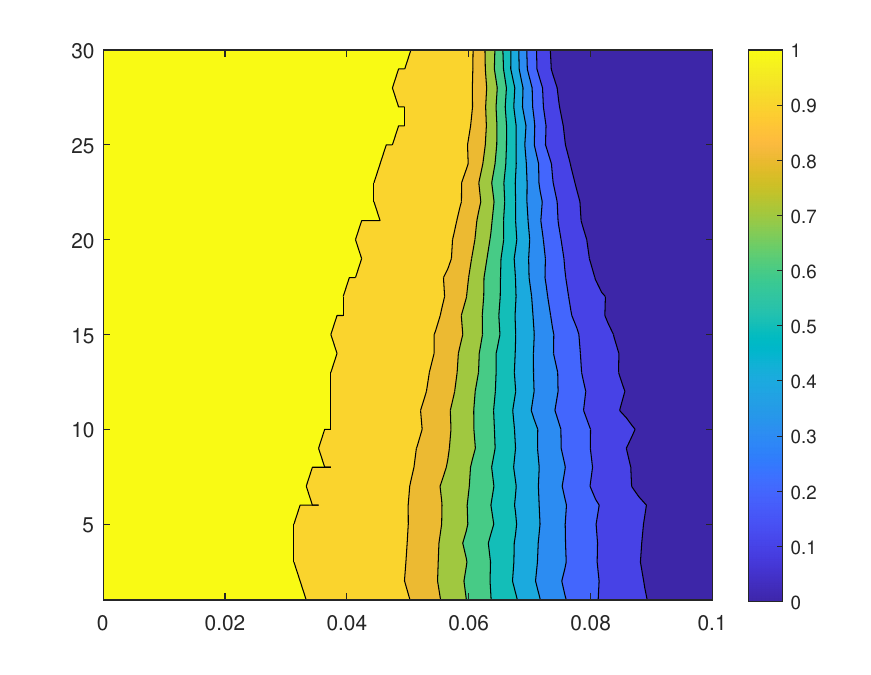}
         \caption{$T=500, R=2000$}
         \label{femproT500R2000}
     \end{subfigure}
     \hfill
     \begin{subfigure}[c]{0.3\textwidth}
         \centering
         \includegraphics[width=\textwidth]{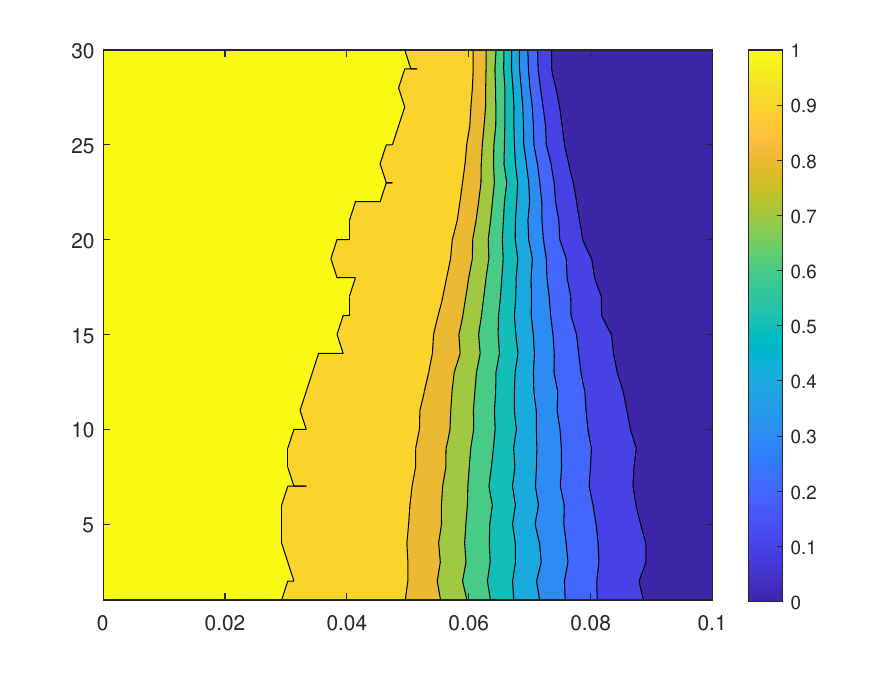}
         \caption{$T=500, R=5000$}
         \label{femproT500R5000}
     \end{subfigure}
        \begin{subfigure}[c]{0.3\textwidth}
         \centering
         \includegraphics[width=\textwidth]{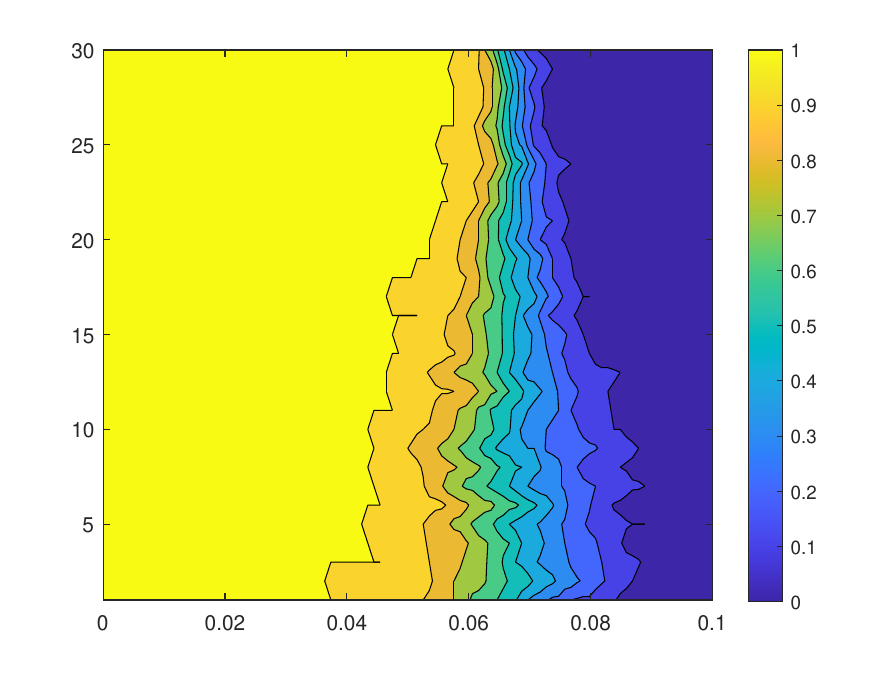}
         \caption{$T=1000, R=100$}
         \label{femproT1000R100}
     \end{subfigure}
     \hfill
     \begin{subfigure}[c]{0.3\textwidth}
         \centering
         \includegraphics[width=\textwidth]{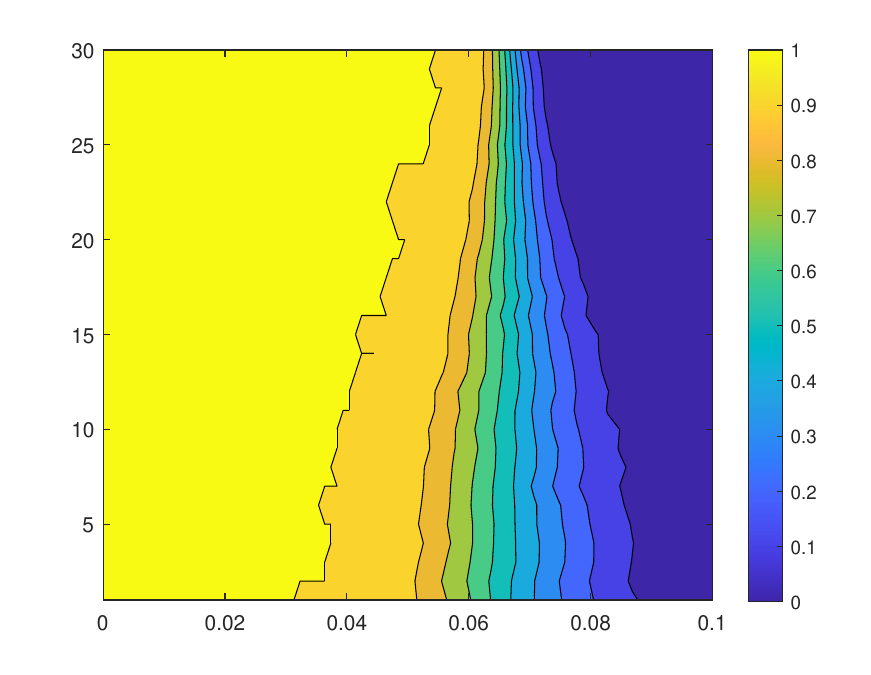}
         \caption{$T=1000, R=2000$}
         \label{femproT1000R2000}
     \end{subfigure}
     \hfill
     \begin{subfigure}[c]{0.3\textwidth}
         \centering
         \includegraphics[width=\textwidth]{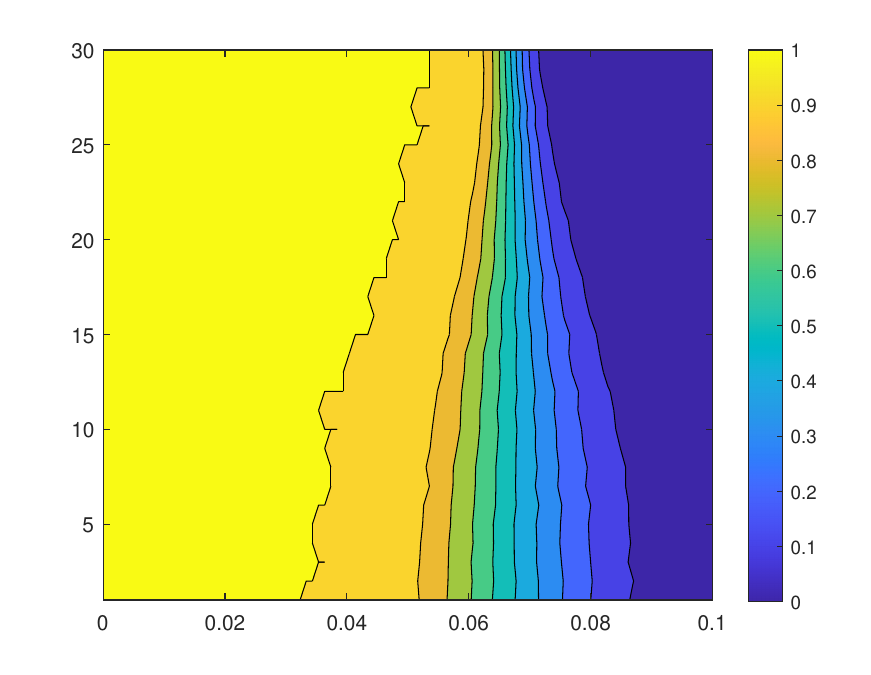}
           \caption{$T=1000, R=5000$}
         \label{femproT1000R5000}
     \end{subfigure}
      \hfill
     \begin{subfigure}[t]{0.3\textwidth}
         \centering
         \includegraphics[width=\textwidth]{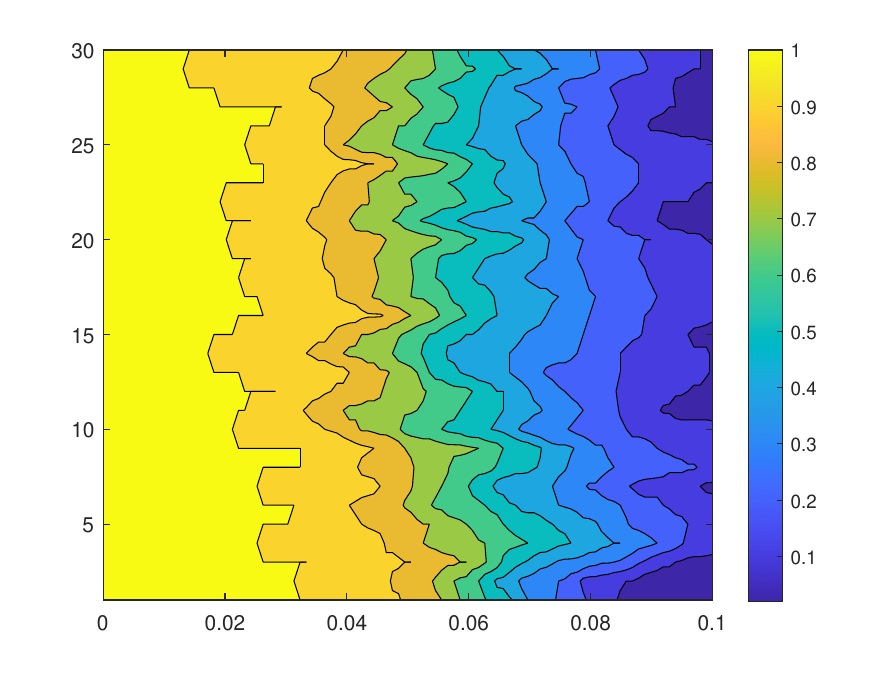}
         \caption{$T=50, R=100$}
         \label{femproT50R100NCsphar3}
     \end{subfigure}
     \hfill
     \begin{subfigure}[t]{0.3\textwidth}
         \centering
         \includegraphics[width=\textwidth]{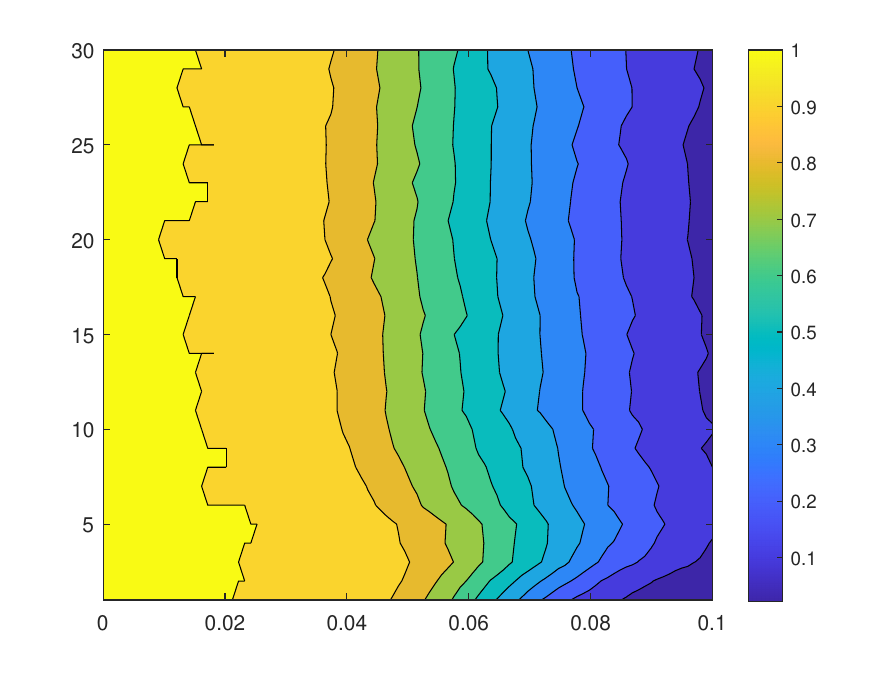}
         \caption{$T=50, R=2000$}
         \label{femproT50R2000NC}
     \end{subfigure}
     \hfill
     \begin{subfigure}[t]{0.3\textwidth}
         \centering
         \includegraphics[width=\textwidth]{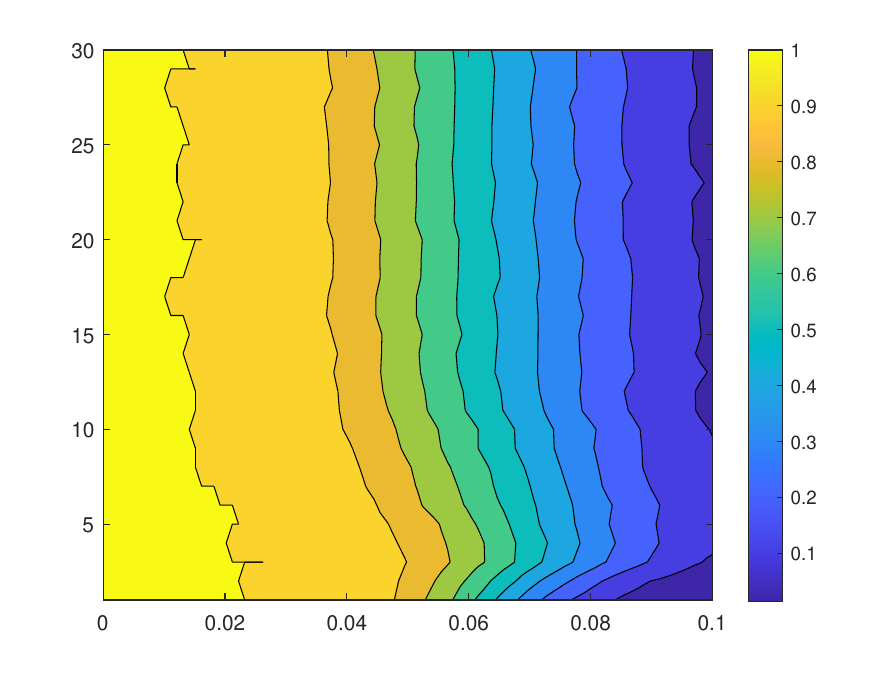}
         \caption{$T=50, R=5000$}
         \label{femproT50R5000NC}
     \end{subfigure}
     \begin{subfigure}[c]{0.3\textwidth}
         \centering
         \includegraphics[width=\textwidth]{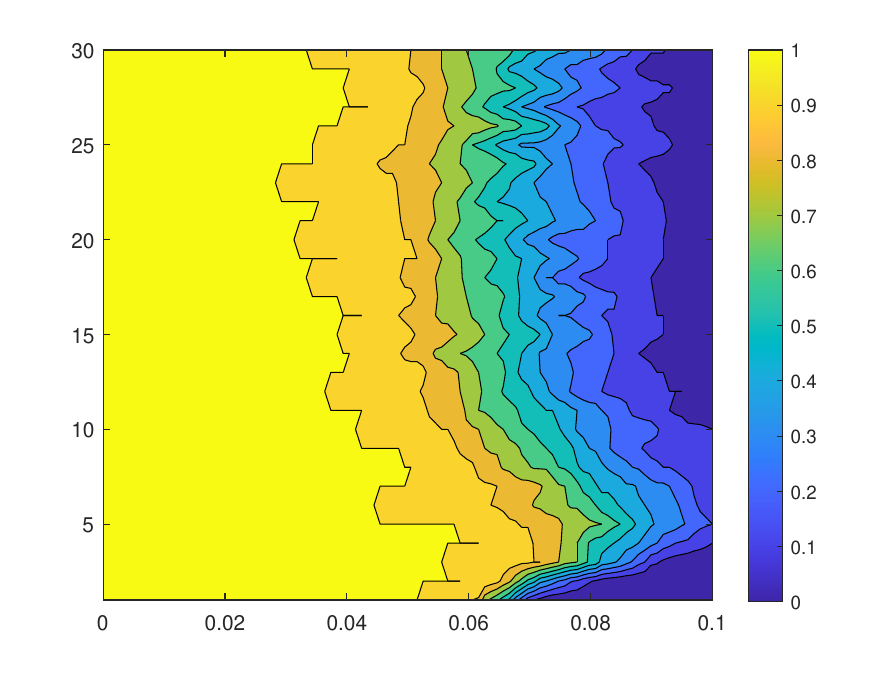}
         \caption{$T=500, R=100$}
         \label{femproT500R100NC}
     \end{subfigure}
     \hfill
     \begin{subfigure}[c]{0.3\textwidth}
         \centering
         \includegraphics[width=\textwidth]{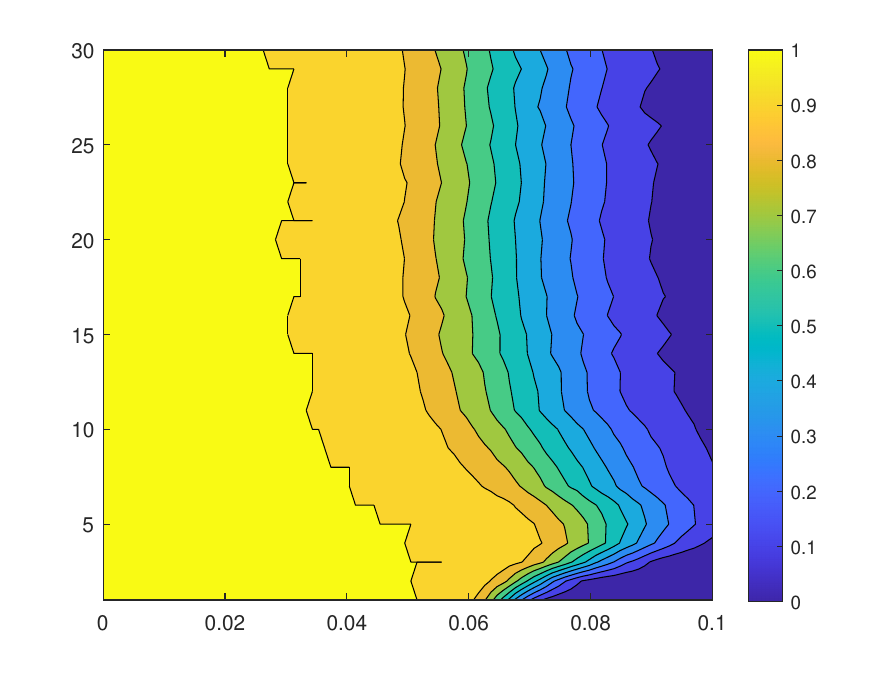}
         \caption{$T=500, R=2000$}
         \label{femproT500R2000NC}
     \end{subfigure}
     \hfill
     \begin{subfigure}[c]{0.3\textwidth}
         \centering
         \includegraphics[width=\textwidth]{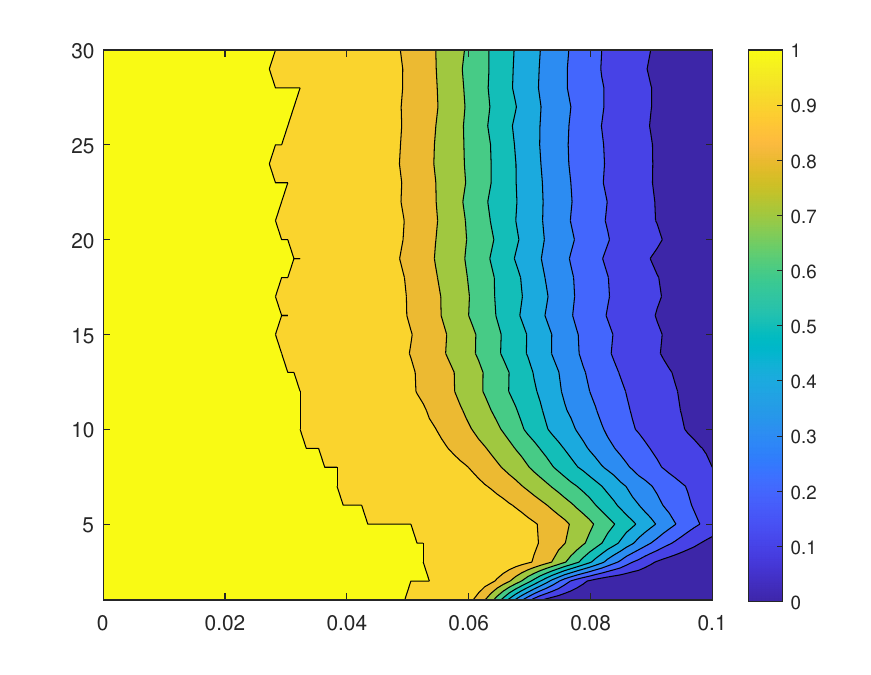}
         \caption{$T=500, R=5000$}
         \label{femproT500R5000NC}
     \end{subfigure}
        \begin{subfigure}[c]{0.3\textwidth}
         \centering
         \includegraphics[width=\textwidth]{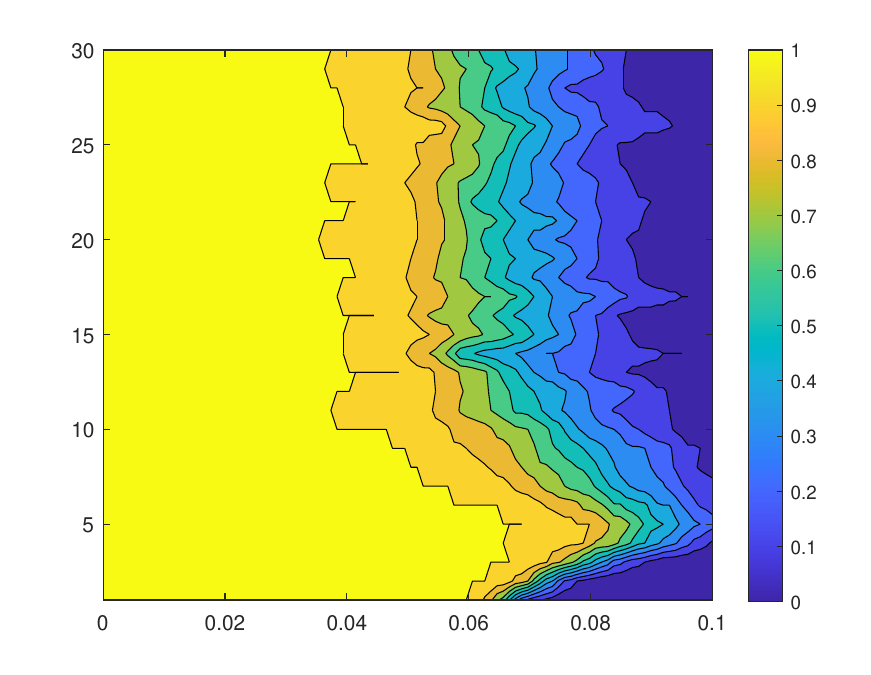}
         \caption{$T=1000, R=100$}
         \label{femproT1000R100NC}
     \end{subfigure}
     \hfill
     \begin{subfigure}[c]{0.3\textwidth}
         \centering
         \includegraphics[width=\textwidth]{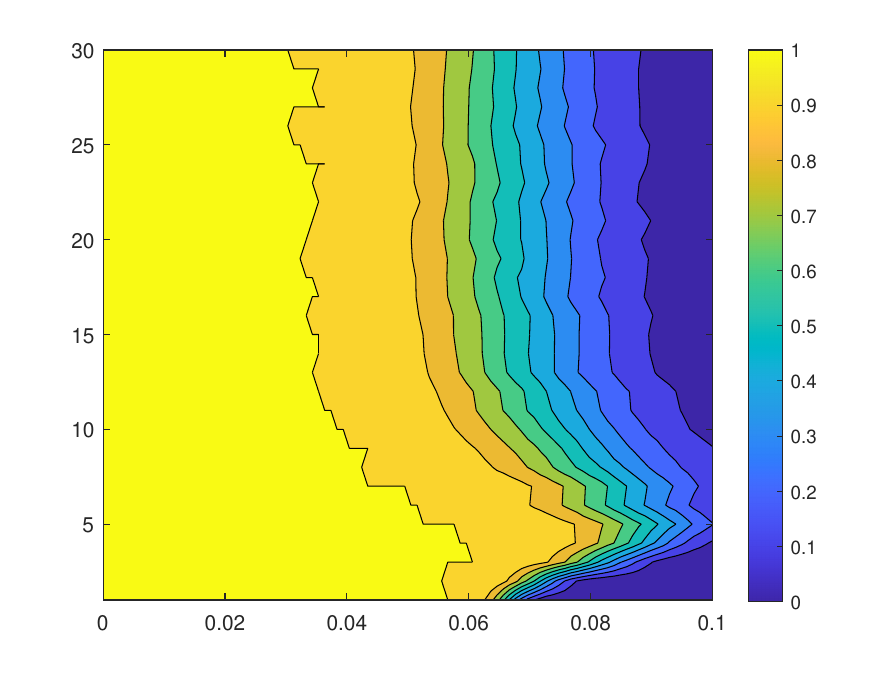}
         \caption{$T=1000, R=2000$}
         \label{femproT1000R2000NC}
     \end{subfigure}
     \hfill
     \begin{subfigure}[c]{0.3\textwidth}
         \centering
         \includegraphics[width=\textwidth]{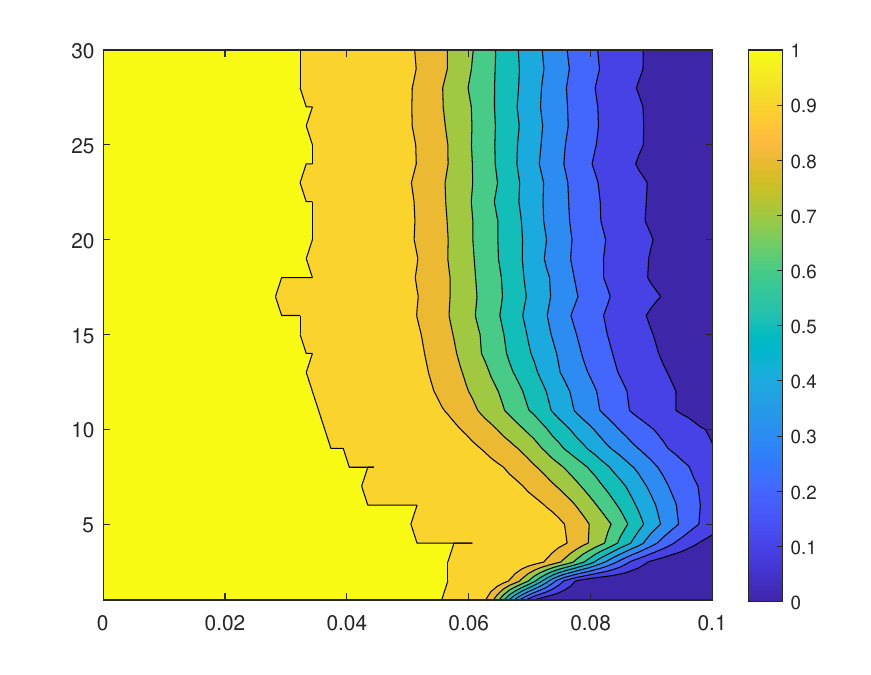}
         \caption{$T=1000, R=5000$}
         \label{femproT1000R5000NC}
     \end{subfigure}
             \caption{Empirical probabilities in equation (\ref{ep}), LRD operator decreasing eigenvalue sequence (plots (a)--(i)), and LRD operator increasing eigenvalue sequence  (plots (j)--(r))}
                \label{femprobsph3}
\end{figure}

\subsection{SRD--LRD estimation results}
\label{sim2}
In this section,  we study the case where  the  projected spherical functional  process displays LRD   at discrete Legendre frequencies $n=1,\dots,15,$ while   SRD is observed at discrete Legendre  frequencies $n=16,\dots,30.$
The eigenvalues
$\{ \alpha (n,\theta ),\ n=1,\dots,15 \}$ of   $\mathcal{A}_{\theta }$ are  displayed for $\theta =\theta_{0}$   in  Figure \ref{eqlrdeigex2}.  The   $100$ candidate  eigenvalue  systems    $\{ \alpha (n,\theta_{i} ), \ n=1,\dots,15, \ i=1,\dots, 100\},$ $\theta_{i}\in \Theta,$ $i=1,\dots, 100,$   involved  in the implementation of the  minimum contrast estimation procedure  are also showed in Figure \ref{eqlrdeigex2}. These  eigenvalue systems are computed from sampling a scaled beta distribution with parameters $2$ and  $\frac{5i}{i+1},$  $i=1,\dots,100.$
\begin{figure}[H]
\begin{center}
\includegraphics[height=3.5cm, width=8cm]{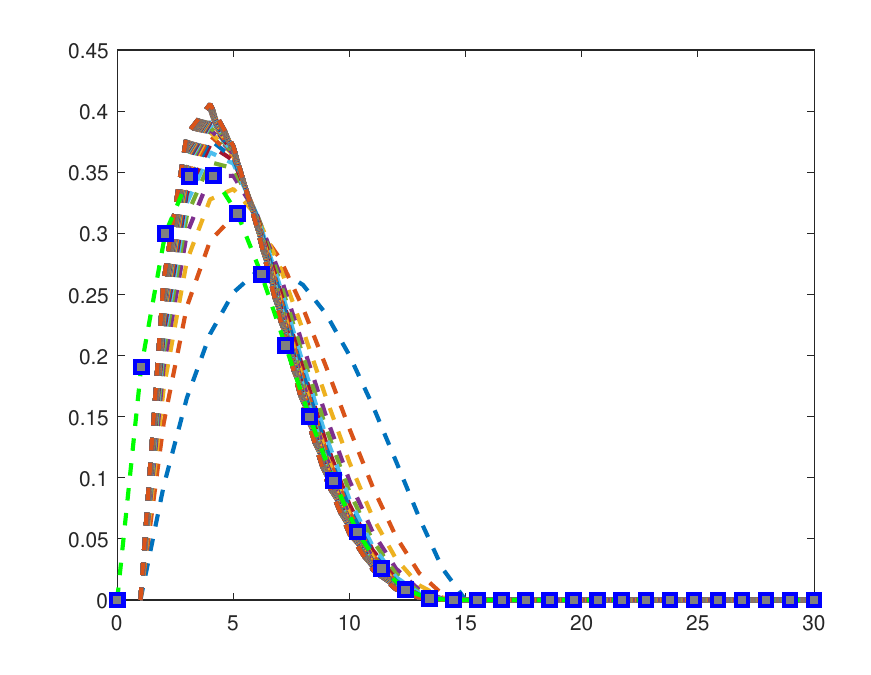}
\caption{ \scriptsize{Eigenvalues of $\mathcal{A}_{\theta_{0}}$ for the case of SRD--LRD  models (dotted blue--green line), and the eigenvalues of the  $100$ parametric candidates $\mathcal{A}_{\theta_{i}},$ $i=1,\dots,100,$ considered }}\label{eqlrdeigex2}
\end{center}
\end{figure}
For multifractionally integrated SPHARMA(1,1) process,  the modulus  of the  fDFT projected    into the eigenpaces $H_{n},$ $n=1,\dots, 15,$ of the spherical Laplace Beltrami operator,  where LRD is displayed, is showed at the left--hand side of Figure \ref{figregfDFTex2}. The  modulus of the  fDFT projected into the eigenspaces $H_{n},$   $n=16,\dots,30,$ where SRD is displayed, can be found at the right--hand side of Figure \ref{figregfDFTex2}.  In the last case, its  weighted kernel operator estimator, based on the Gaussian kernel and bandwidth parameter $B_{T}=0.2,$ is  given at the right--hand side of
Figure \ref{weightedperestex2bb}. For the minimum contrast parametric estimator  $\widehat{f}_{\omega  }^{(T,n_{0})}(\cdot ,\widehat{\theta }_{T})$ with $n_{0}=15,$    minimizing  over the $100$ candidate parametric models
  the  bounded linear operator norm $\mathcal{L}(L^{2}(\mathbb{S}_{2}, d\nu; \mathbb{C}))$  of the  empirical contrast operator $\mathbf{U}_{T,\theta }$ in  (\ref{eco}),  Figure \ref{EX2MCOLSFPE} displays the empirical probabilities $\widehat{\mathcal{P}}\left( \|f_{n}(\cdot )-\widehat{f}^{(T,(n_{0}))}_{n}(\cdot  ,\widehat{\theta}_{T}) \|_{L^{1}([-\pi,\pi])}>\varepsilon_{i}\right),$
 for $n=1,\dots,15,$ and thresholds
$\varepsilon_{i}=i(0.016) \in (0,0.8),$ $i=1,\dots,50.$

The performance of the SRD--LRD estimation methodology is also analyzed  for the remaining multifractionally integrated spherical functional processes generated as displayed in  Figure \ref{fmixedT500R2000}. Specifically, in this figure,   results in terms of the  empirical mean quadratic errors, associated with SRD functional spectral  estimation (left hand--side),  and in terms of the histograms of the  temporal mean of the empirical  absolute errors  (right hand--side), associated with LRD spectral estimation, are respectively showed from $R=300$ independent  generations of a  functional sample of  size $T=500.$   Results  for    $T=50,$  $T=100,$  $T=500,$    $T=1000,$ and $R=100$ are displayed in  Section 5 of the Supplementary Material.
\begin{figure}[!htb]
\begin{center}
\includegraphics[height=0.159\textheight, width=0.4\textwidth]{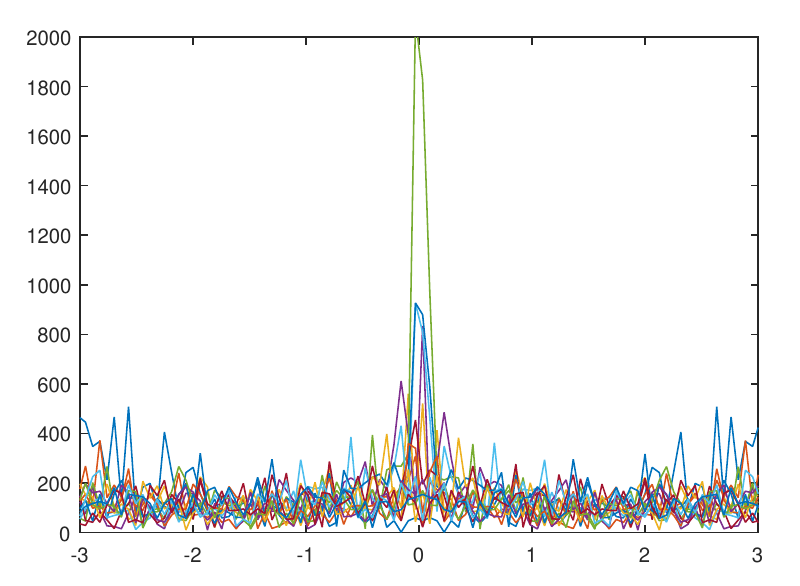}
\includegraphics[height=0.159\textheight, width=0.4\textwidth]{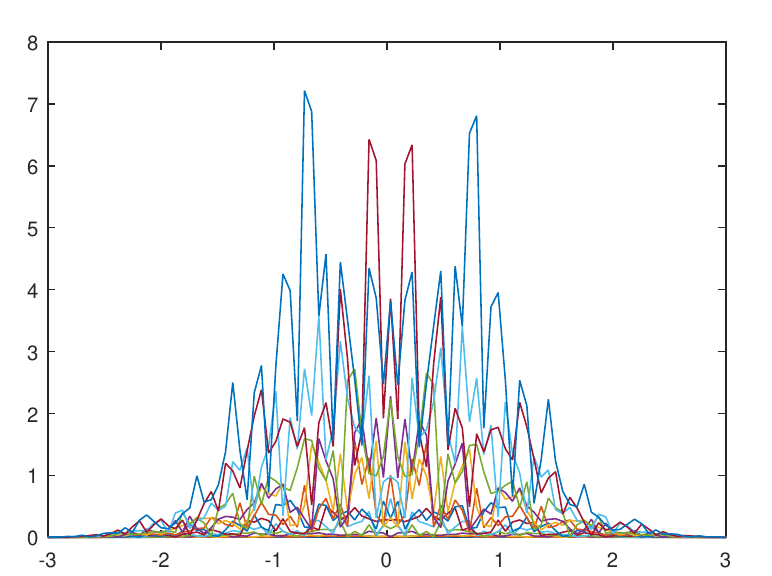}
\caption{ \scriptsize{Modulus of the projected   LRD fDFT  for  discrete Legendre frequencies $n=1,\dots,15$ (left--hand side), and  modulus of the projected   SRD fDFT  for  discrete Legendre frequencies $n=16,\dots,30$ (right--hand side) }}\label{figregfDFTex2}
\end{center}
\end{figure}

 \begin{figure}[!htb]
\begin{center}
\includegraphics[height=0.159\textheight, width=0.4\textwidth]{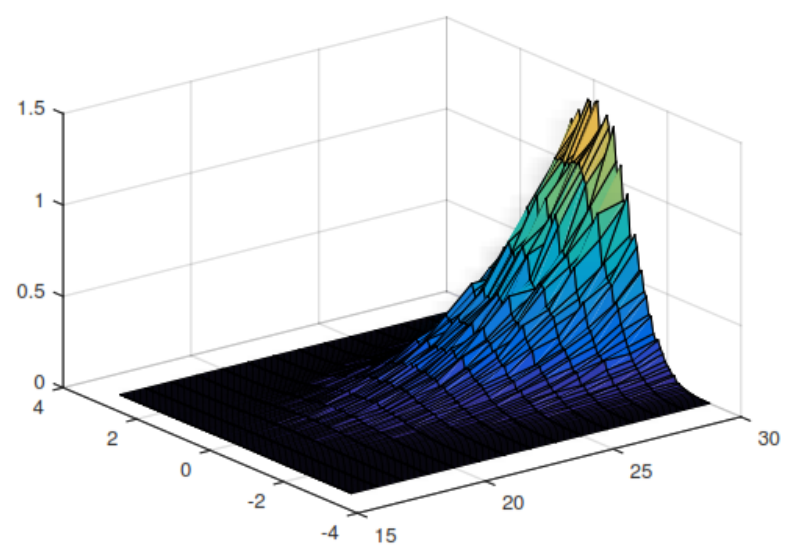}
\includegraphics[height=0.159\textheight, width=0.4\textwidth]{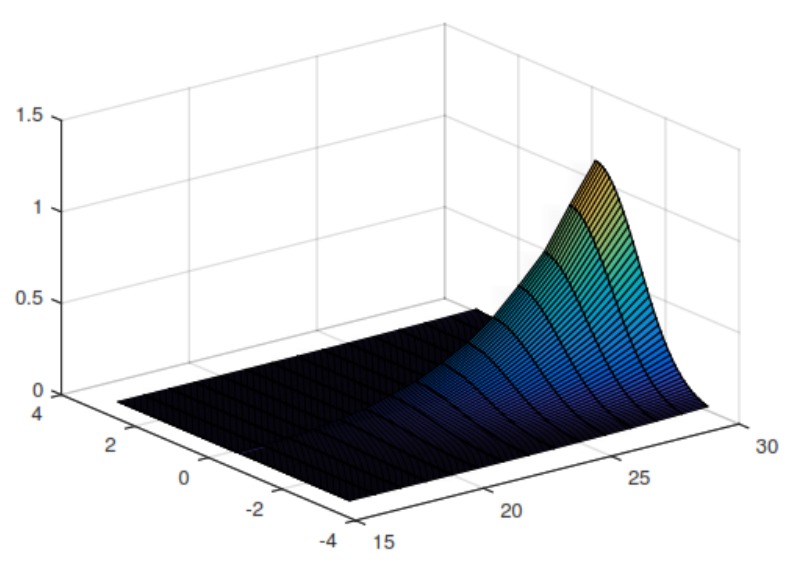}
 \caption{ \scriptsize{ The modulus of the  fDFT projected into the eigenspaces $H_{n},$   $n=16,\dots,30,$ (left--hand side),  and  its  weighted kernel operator estimator (right--hand side)
}}\label{weightedperestex2bb}
\end{center}
\end{figure}

\begin{figure}[!htb]
\begin{center}
\includegraphics[height=0.159\textheight, width=0.4\textwidth]{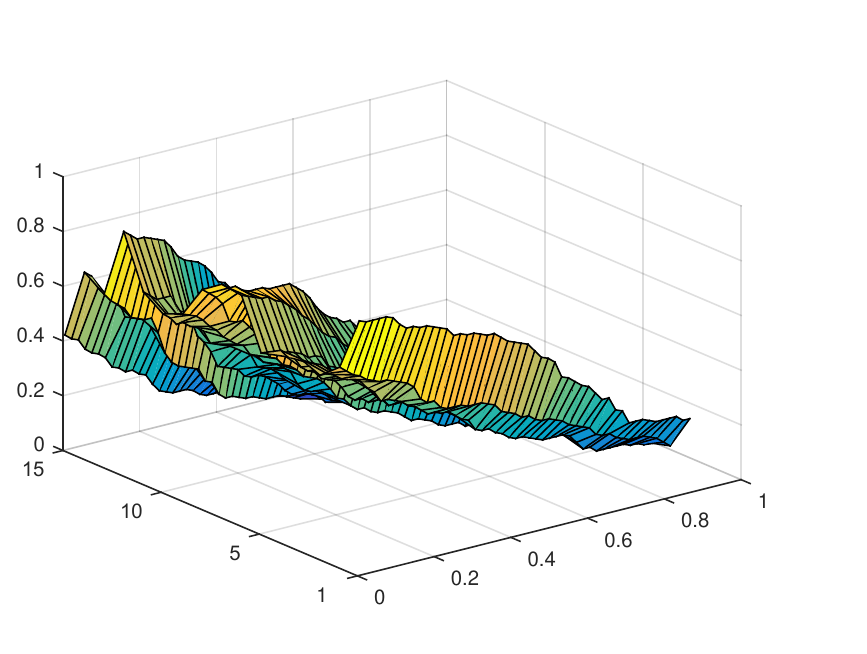}
\includegraphics[height=0.159\textheight, width=0.4\textwidth]{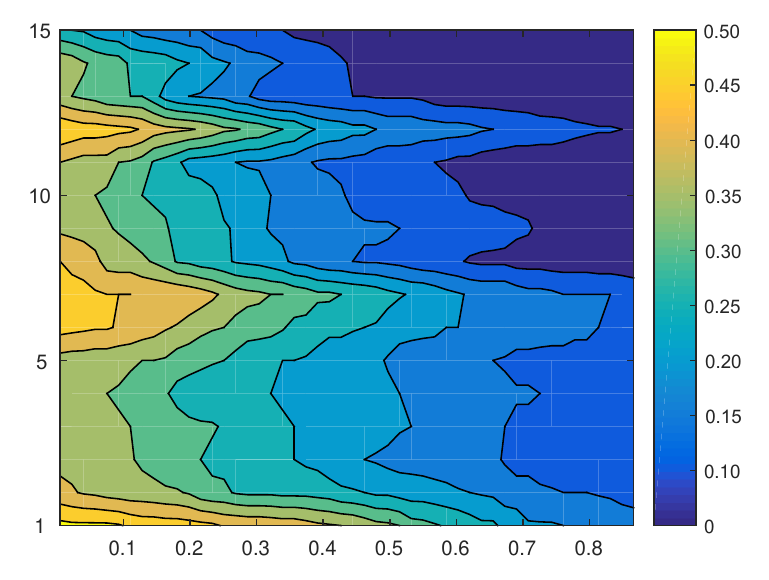}
\caption{ \scriptsize{The empirical probabilities $\widehat{\mathcal{P}}\left( \left\|f_{n}(\cdot )-\widehat{f}_{n}(\cdot ,\widehat{\theta}_{T})^{(T,n_{0})} \right\|_{L^{1}([-\pi,\pi])}>\varepsilon_{i}\right),$ for $n=1,\dots,15,$ and
$\varepsilon_{i}=i(0.016) \in (0,0.8),$ $i=1,\dots,50,$
based on  $100$ independent generations of a functional sample of size  $T=100$
(surface  and contour plot left to right--hand side)}}\label{EX2MCOLSFPE}
\end{center}
\end{figure}
\setcounter{section}{5} \setcounter{equation}{0} 

\begin{figure}[H]
     \centering
     \begin{subfigure}[t]{0.3\textwidth}
         \centering
         \includegraphics[height=0.15\textheight, width=\textwidth]{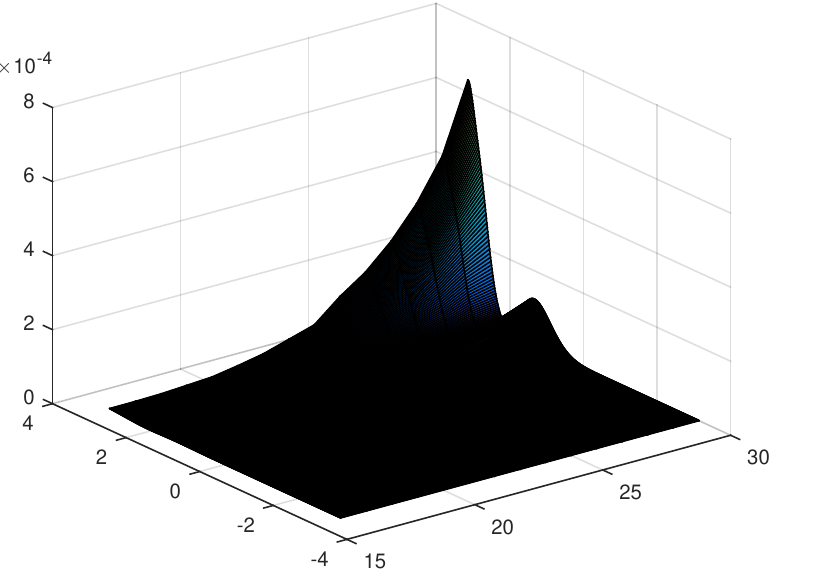}
         \caption{SPHAR(1). E.M.Q.E. of P.W.P.O.}
         \label{fEQMIXEDSPHAR1T500R2000SPHAR1}
     \end{subfigure}
     \hfill
     \begin{subfigure}[t]{0.3\textwidth}
         \centering
         \includegraphics[height=0.15\textheight, width=\textwidth]{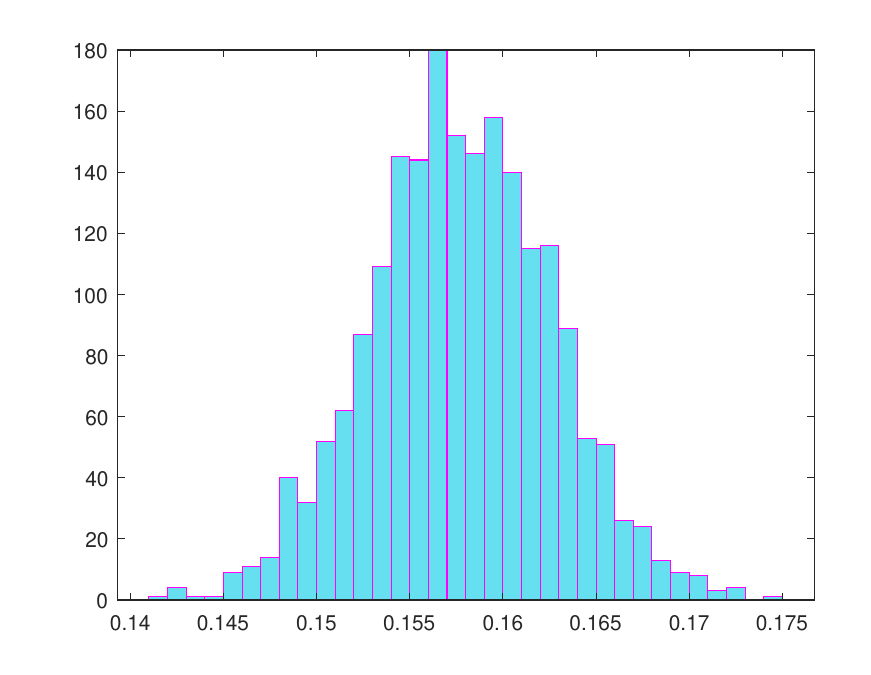}
         \caption{M.I. SPHAR(1). Scale $10$. T.M.E.A.E. histograms}
         \label{fMIXEDEQES15T500R2000SPHAR1}
     \end{subfigure}
     \hfill
     \begin{subfigure}[t]{0.3\textwidth}
         \centering
         \includegraphics[height=0.15\textheight, width=\textwidth]{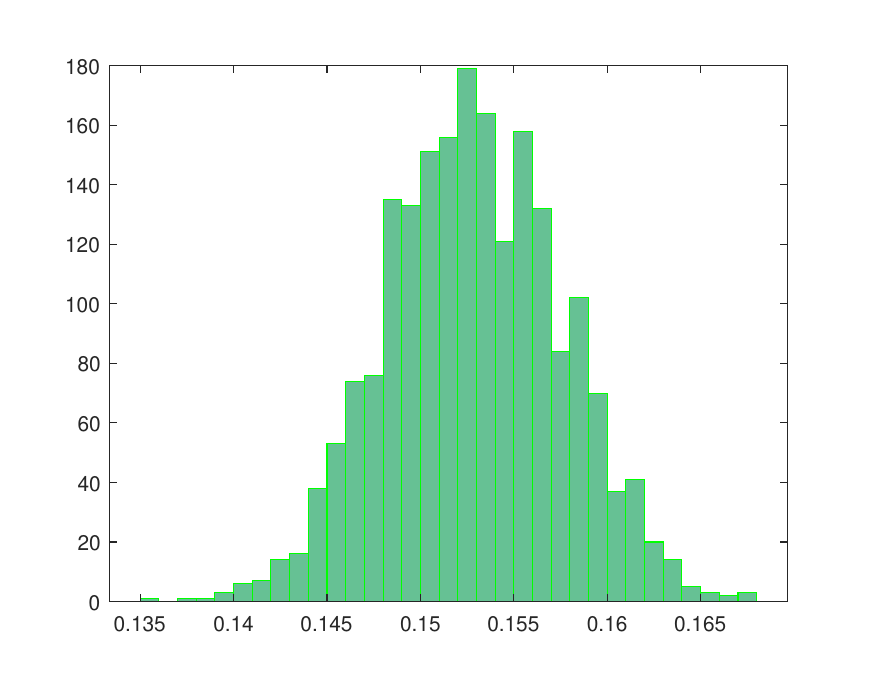}
         \caption{M.I. SPHAR(1). Scale $15$. T.M.E.A.E. histograms}
         \label{fMIXEDEQES1015T500R2000SPHAR1}
     \end{subfigure}
          \begin{subfigure}[c]{0.3\textwidth}
         \centering
         \includegraphics[height=0.15\textheight, width=\textwidth]{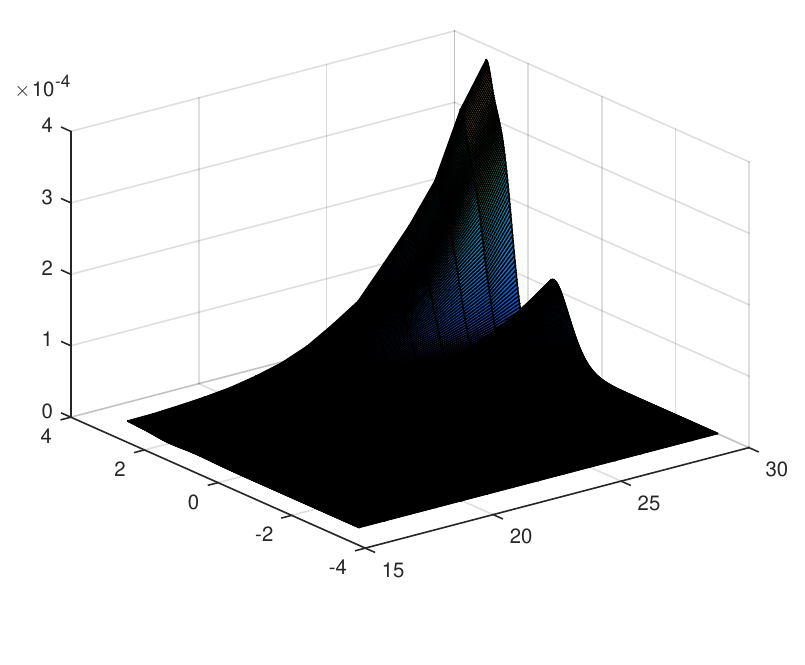}
         \caption{SPHAR(3). E.M.Q.E. of P.W.P.O.}
         \label{fEQMIXEDSPHAR3T500R2000SPHAR3}
     \end{subfigure}
     \hfill
     \begin{subfigure}[c]{0.3\textwidth}
         \centering
         \includegraphics[height=0.15\textheight, width=\textwidth]{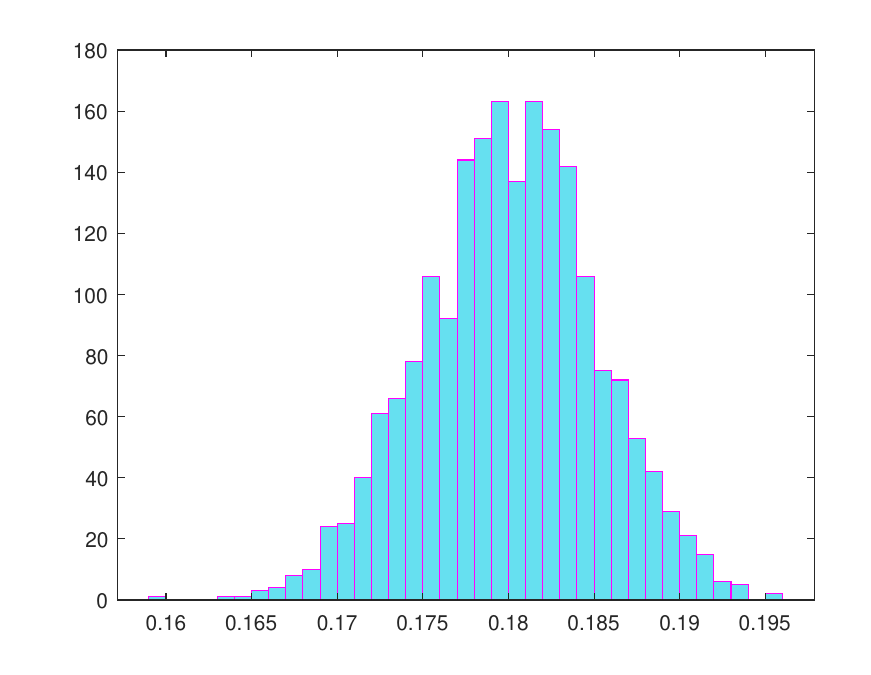}
         \caption{M.I. SPHAR(3). Scale  $10$. T.M.E.A.E. histograms}
         \label{fMIXEDEQES15T500R2000SPHAR3}
     \end{subfigure}
     \hfill
     \begin{subfigure}[c]{0.3\textwidth}
         \centering
         \includegraphics[height=0.15\textheight, width=\textwidth]{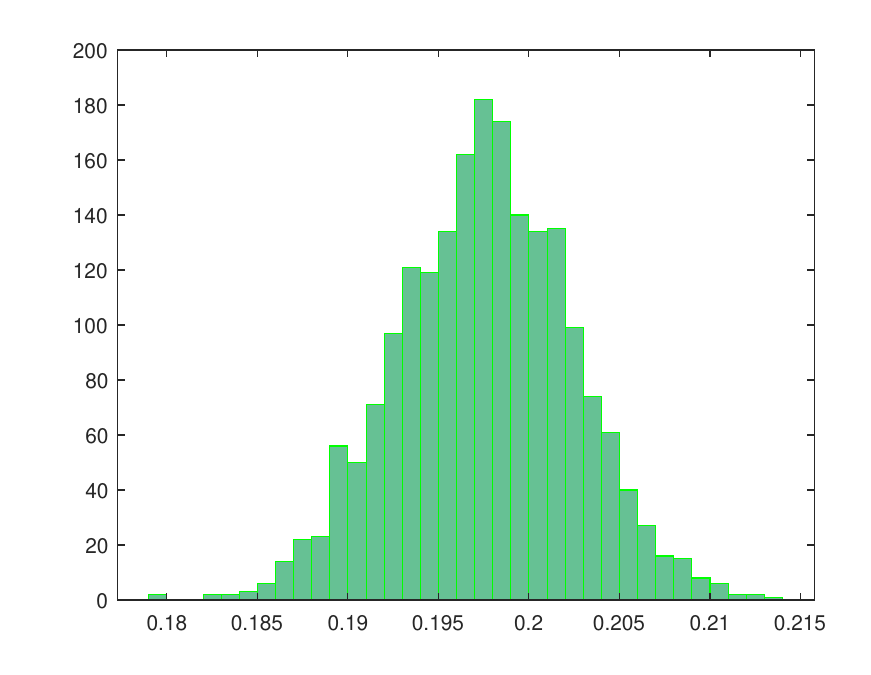}
         \caption{M.I. SPHAR(3). Scale  $15$. T.M.E.A.E. histograms }
         \label{fMIXEDEQES1015T500R2000SPHAR3}
     \end{subfigure}
         \begin{subfigure}[C]{0.3\textwidth}
         \centering
         \includegraphics[height=0.15\textheight, width=\textwidth]{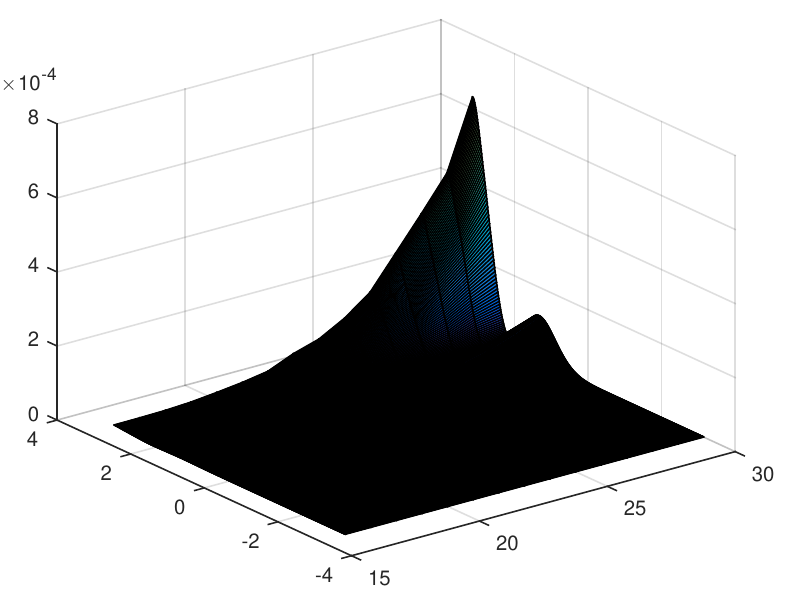}
         \caption{SPHARMA(1,1). E.M.Q.E. of P.W.P.O.}
         \label{fEQMIXEDSPHARMA11T500R2000SPHARMA11}
     \end{subfigure}
     \hfill
     \begin{subfigure}[C]{0.3\textwidth}
         \centering
         \includegraphics[height=0.15\textheight, width=\textwidth]{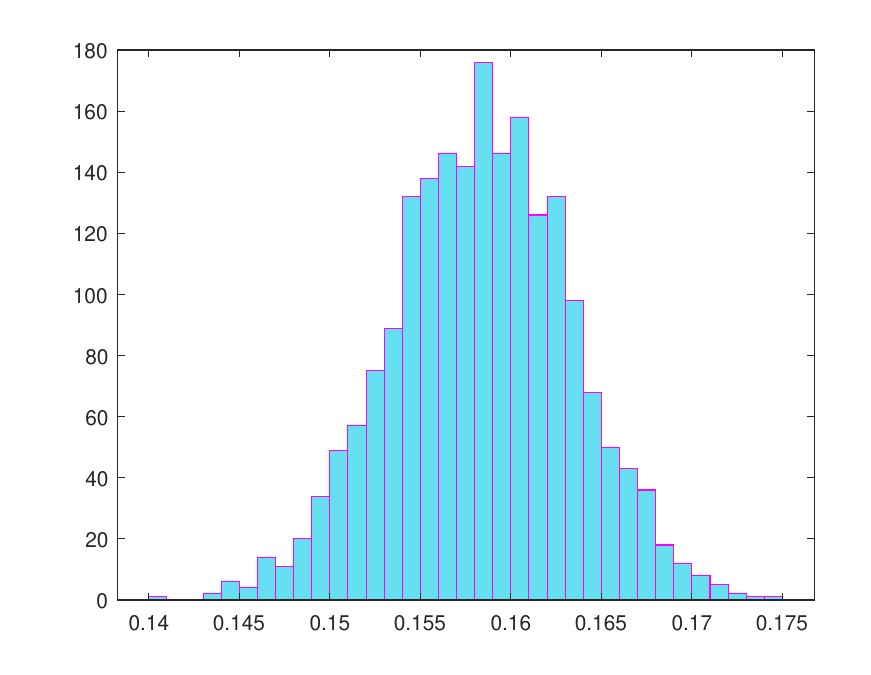}
         \caption{M.I. SPHARMA(1,1). Scale $10$. T.M.E.A.E. histograms}
         \label{fMIXEDEQES15T500R2000SPHARMA11}
     \end{subfigure}
     \hfill
     \begin{subfigure}[C]{0.3\textwidth}
         \centering
         \includegraphics[height=0.15\textheight, width=\textwidth]{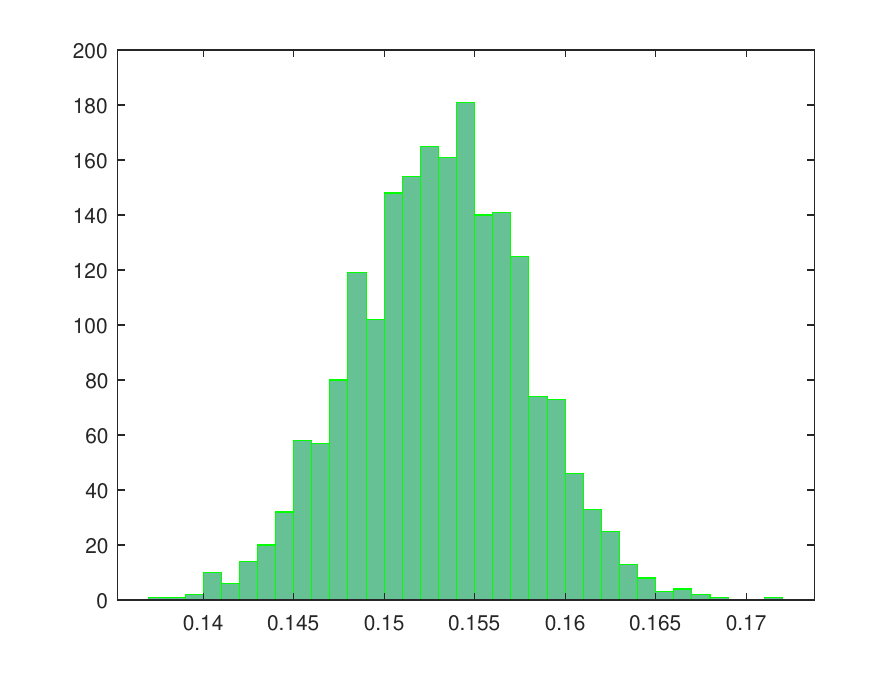}
         \caption{M.I. SPHARMA(1,1). Scale $15$. T.M.E.A.E. histograms}
         \label{fMIXEDEQES1015T500R2000SPHARMA11}
     \end{subfigure}
          \begin{subfigure}[t]{0.3\textwidth}
         \centering
         \includegraphics[height=0.15\textheight, width=\textwidth]{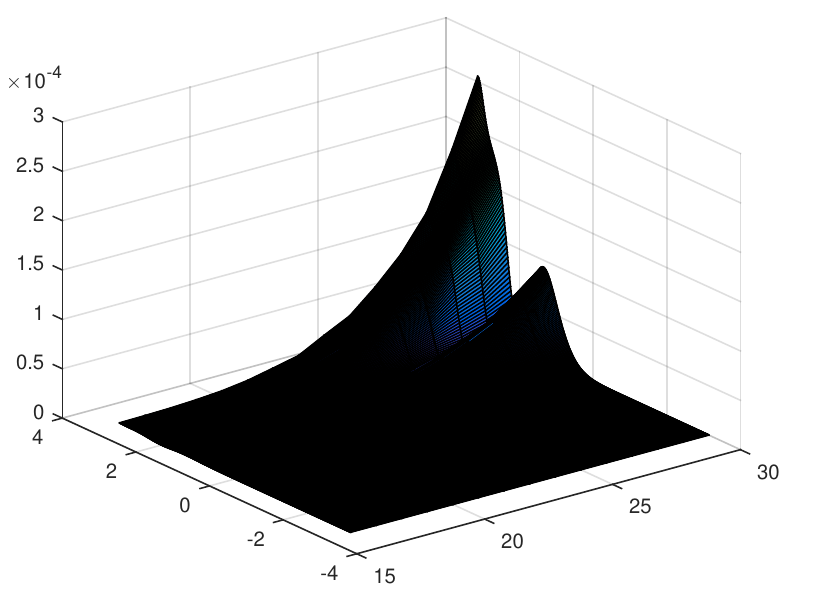}
         \caption{SPHARMA(3,1). E.M.Q.E. of P.W.P.O.}
         \label{fEQMIXEDSPHARMA31T500R2000SPHARMA31}
     \end{subfigure}
     \hfill
     \begin{subfigure}[t]{0.3\textwidth}
         \centering
         \includegraphics[height=0.15\textheight, width=\textwidth]{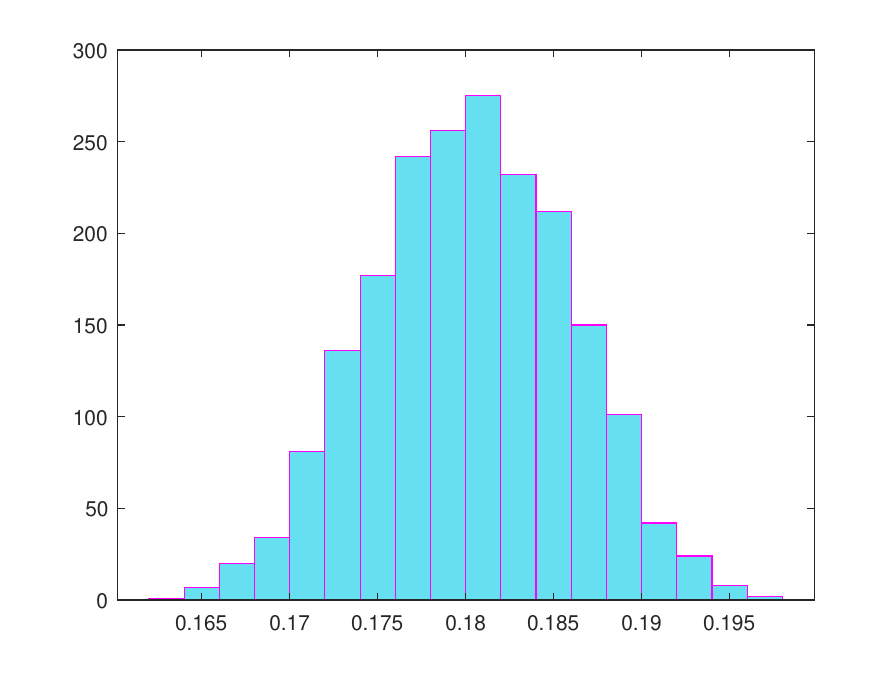}
         \caption{M.I. SPHARMA(3,1). Scale $10$. T.M.E.A.E. histograms }
         \label{fMIXEDEQES15T500R2000SPHARMA31}
     \end{subfigure}
     \hfill
     \begin{subfigure}[t]{0.3\textwidth}
         \centering
         \includegraphics[height=0.15\textheight, width=\textwidth]{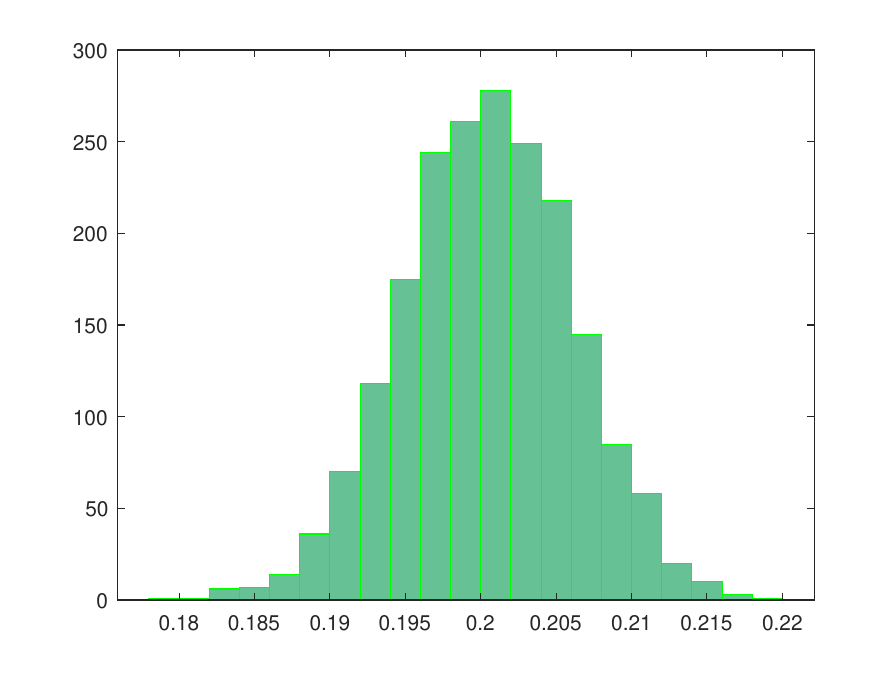}
         \caption{M.I. SPHARMA(3,1). Scale $15$. T.M.E.A.E. histograms}
         \label{fMIXEDEQES1015T500R2000SPHARMA31}
     \end{subfigure}
        \caption{\scriptsize{The empirical mean quadratic errors (E.M.Q.E.s)  associated with the projected weighted periodogram  operator (P.W.P.O.) estimator (eigenspaces $H_{n},$ $n=16,\dots, 30,$ of the  Laplace Beltrami operator), are displayed at the left--hand side. The remaining plots provide the histograms of the temporal mean of the empirical absolute errors, associated with the minimum contrast parameter estimation of LRD operator defining multifractional integration (M.I),  for eigenspaces $H_{n},$  $n=10,15.$ All the results displayed are based on $R=300$ independent generations of a functional sample of size  $T=500.$  The bandwidth parameter $B_{T}=0.2$ has been chosen in all the cases  }}
        \label{fmixedT500R2000}
\end{figure}
One can observe since  $T=500$ the order of magnitude of the empirical mean quadratic errors is $10^{-4}$ for the bandwidth parameter $B_{T}=0.2$ in the SRD estimation.  It is well-known that the bandwidth parameter affects precision of the weighted periodogram operator estimator, and the impact of parameter
$T$ is very strong. This fact can also be observed in Section 5 of the Supplementary Material,  looking at  differences in the magnitude of empirical mean quadratic errors, based on $100$ repetitions, associated with the weighted periodogram operator for $T=50$ and $T=100,$ considering bandwidth parameter $B_{T}= 0.1,$ as well as for $T=500,$ and $T=1000,$ considering bandwidth parameter $B_{T}= 0.2,$ since a substantial reduction in such magnitudes occurs when increasing the functional sample size $T$ (see, e.g., Theorem 3.6 in \cite{Panaretos13}).  Similar  results as in previous section are obtained in terms of  the empirical distribution of the temporal mean of the empirical  absolute errors in the implementation of minimum contrast estimation methodology for $n=1,\dots,15.$
\section{Final comments}\label{sec:7}

Results displayed in   Sections \ref{mcess} and  \ref{sim2}  (see also Sections 4 and 5 of the Supplementary Material) are based on the computation of the empirical distribution of the temporal mean of the absolute errors, and the empirical probabilities.
  The interaction between  parameters   $n,$  $R$  and $T$ in the asymptotic analysis of the two estimation methodologies proposed is illustrated  beyond the Gaussian scenario in the simulation study undertaken.  Specifically, from the empirical distributions plotted, one can conclude that  their  rate of convergence is a function of  the spherical scale $n,$ the functional sample size $T,$ and the number of repetitions $R.$ Differences between empirical distributions of absolute errors by scales are more pronounced for decreasing sequence of LRD operator eigenvalues  than in the case of increasing LRD operator eigenvalue sequence.
  As expected  (see also Supplementary Material), the effect of the element of  SPHARMA(p,q) process  family  considered is negligible for the minimum contrast estimation methodology. Namely, a slightly  increasing  of the concentration   rate of the empirical errors  when   the parameter of autoregression $p$ increases is observed.
  Additionally, the empirical probability analysis also reflects the interaction of these three parameters through  the  rate of converge to zero. Under nondecreasing  eigenvalue sequence of the LRD operator, a smoother decay to zero of the empirical probabilities than in the  case of decreasing eigenvalue sequence of the LRD operator  is observed.
 A new battery of limit results  will be investigated beyond the Gaussian scenario in a subsequent paper for  the asymptotic analysis of the proposed estimators of the second--order structure of the  LRD manifold cross--time  RFs     studied here.

\setcounter{section}{7} \setcounter{equation}{0} 
\begin{acknowledgements}
  This work has been supported in part by projects
MCIN/ AEI/PID2022-142900NB-I00,  and CEX2020-001105-M MCIN/ AEI/10.13039/501100011033).

The authors would like to thank the Editor in Chief and Associate Editor, as well as  two anonymous reviewers for their helpful and constructive comments and suggestions which have led to a substantial  improvement of  this paper.
 \end{acknowledgements}
 {\small
 The authors declare that they have no conflict of interest.}



\end{document}